\documentclass[twocolumn]{autart}
\usepackage{amsmath} % assumes amsmath package installed
\usepackage{amssymb}  % assumes amsmath package installed

\usepackage{bm}
\usepackage{xcolor}
\usepackage{comment}
\usepackage[font=small]{caption}
\usepackage{graphicx} % for pdf, bitmapped graphics files
\usepackage[font=small]{subfig}
\usepackage{siunitx}
\newtheorem{theorem}{Theorem}[section]
\newtheorem{proposition}{Proposition}[section]
\newtheorem{lemma}{Lemma}[section]
\newtheorem{assumption}{Assumption}[section]
\newtheorem{corollary}{Corollary}[section]
\newtheorem{definition}{Definition}[section]
\newtheorem{remark}{Remark}[section]

\newenvironment{proof}
     {\textit{Proof:}}{\hfill $\blacksquare$\\}
\makeatletter
\let\c@author\relax
\makeatother
\usepackage[maxbibnames=6,giveninits=true,sorting=none]{biblatex} %Imports biblatex package

\AtEveryBibitem{%
  \clearfield{issn} % Remove issn
  \clearfield{doi} % Remove doi
  \clearfield{url}
   \clearfield{number}
}
\DeclareFieldFormat*{year}{{}}
\renewbibmacro{in:}{} %removes In: in bibliography
\newcommand{\tr}[1]{{#1}^{\ensuremath{\mathsf{T}}}} % transpose 
\newcommand{\inv}[1]{{#1}^{\ensuremath{\mathsf{-1}}}} % inverse

\usepackage{xcolor}
\newcommand{\mxx}[1]{{\color{blue}#1\ }}  % Naveed's comments
\newcommand{\todo}[1]{{\color{orange}#1\ }}  % TODO

\begin{document}
\begin{frontmatter}
\title{An Information-State Based Approach to Linear Time Varying System Identification and Control}

%\thispagestyle{plain}
%\pagestyle{plain}
%\setcounter{page}{44}

%\author{Mohamed Naveed Gul Mohamed, Raman Goyal, Suman Chakravorty, and Ran Wang% <-this % stops a space
%\thanks{The authors are with the Department of Aerospace Engineering, Texas A\&M University, College Station, TX 77843 USA. \{\tt naveed, ramaniitrgoyal92, schakrav,  rwang0417\}@tamu.edu
%}}

\author[tamu]{Mohamed Naveed Gul Mohamed}\ead{naveed@tamu.edu},  
\author[parc]{Raman Goyal}\ead{ramaniitr.goyal92@gmail.com},
\author[tamu]{Suman Chakravorty\thanksref{cauthor}}\ead{schakrav@tamu.edu},
\author[rockwell]{Ran Wang}\ead{rwang0417@tamu.edu}

\address[tamu]{Texas A\&M University, College Station, Texas, USA} 
\address[parc]{Palo Alto Research Center, part of SRI International, California, USA}
\address[rockwell]{Rockwell Automation, Texas, USA}

\thanks[cauthor]{Corresponding author: Suman Chakravorty, Tel. +979-458-0064. }

\begin{abstract}
This paper considers the problem of system identification for linear time varying systems. We propose a new system realization approach that uses an ``information-state" as the state vector, where the ``information-state" is composed of a finite number of past inputs and outputs. The system identification algorithm uses input-output data to fit an autoregressive moving average model (ARMA) to represent the current output in terms of finite past inputs and outputs. 
This information-state-based approach allows us to directly realize a state-space model using the estimated time varying ARMA paramters  linear time varying (LTV) systems. The paper develops the theoretical foundation for using ARMA parameters-based system representation using only the concept of linear observability, details the reasoning for exact output modeling using only the finite history, and shows that there is no need to separate the free and the forced response for identification. The paper also discusses the implications of using the information-state system for optimal output feedback control and shows that the solution obtained using a suitably posed information state problem is optimal for the original problem. The proposed approach is tested on various different systems, and the performance is compared with state-of-the-art LTV system identification techniques.
% and the simulation results show excellent performance.
\end{abstract}

\end{frontmatter}

\section{Introduction}
In this paper, we consider the system identification of linear systems. The motivation for the work arose while trying to accomplish the data-based optimal control of nonlinear systems where we had introduced a linear time-varying (LTV) system identification technique based on an autoregressive moving average (ARMA) model of the input-output map of an LTV system \cite{Ran_ICRA21}. Nonlinear dynamical systems can be approximated either as a linear time-invariant (LTI) system around an equilibrium or a linear time-varying (LTV) system around a trajectory. System identification may need to be done to design control laws for such systems~\cite{ljung1998system,juang1994applied}. System identification techniques have also been used for parameter identification and for generating reduced-order models
\cite{juang1985eigensystem,shokoohi1987identification}. The technique discussed in this paper, building upon the preliminary idea introduced in \cite{Ran_ICRA21}, uses a system realization that is based on the ``information-state" as the state vector. An ARMA model which can represent the current output in terms of inputs and outputs from $q$ steps in the past, is found by solving a linear regression problem relating the input and output data. Defining the state vector to be the past inputs and outputs, as the information-state, lets us realize a state-space model directly using the estimated time-varying ARMA parameters.\\
%\subsection{Related work:} \label{sec.related_work}

The pioneering work in system identification for LTI systems is the Ho-Kalman realization theory \cite{HOKALMAN} of which the Eigensystem Realization Algorithm (ERA) algorithm is one of the most popular ~\cite{juang1985eigensystem}. Another system identification method, namely, $q$-Markov covariance equivalent realization, generates a stable LTI system model that matches the first ``$q$" Markov parameters of the underlying system and also matches the equivalent steady-state covariance response/ parameters of the identified system \cite{majji2018timeCOVER,king1988generalized}. These algorithms assume stable systems so that the response can be modeled using a finite set of parameters relating the past inputs to the current output (moving-average (MA) model). For lightly damped and marginally stable systems, the length of history to be considered and the parameters to be estimated becomes very long, leading to numerical issues when solving for the parameters. To overcome this issue, the observer Kalman identification algorithm (OKID)~\cite{OKID} uses an ARMA model, rather than an MA model, consisting of past outputs and controls to model the current output. The time-varying counterparts of the ERA and OKID - TV-ERA and TV-OKID - were developed in~\cite{TV_ERA} and~\cite{TV_OKID}, respectively. The identification of time varying linear systems (TV-ERA and TV-OKID) also builds on the earlier work on time-varying discrete time system identification \cite{shokoohi1987identification, VERHAEGEN1995201}. The OKID and TV-OKID explain the usage of an ARMA model to be equivalent to an observer in the loop system, and postulate that the identified observer is a deadbeat observer similar to the work in~\cite{juang1997deadbeat}. \\

All the algorithms mentioned above start by estimating the Markov parameters of the system from input-output data: the ERA type techniques directly from the MA model and the OKID methods from recursively solving a set of linear equations involving the ARMA/ observer Markov parameters. A Hankel matrix is built using the estimated Markov parameters, and a singular value decomposition (SVD) of the Hankel matrix is used to find the system matrices. All the approaches either require the experiments to be performed from zero-initial conditions or wait sufficiently long enough for the initial condition response to die out. The time-varying systems also need an additional step of coordinate transformation, as the realized matrices at different time steps are in different coordinate frames~\cite{TV_ERA}. These Hankel SVD-based approaches for LTV systems also require separate experiments - forced response with zero-initial conditions and a free-response with random initial conditions, where the free response experiments are required to identify the system in the first and last few steps.\\ %since the theory is built on the concept of an observer in the loop, as opposed to the proposed theory of the ``information-state" approach which follows from the concept of observability.
%\mxx{Not sure if we need to talk about them here but maybe we should look into them in future \cite{juang1997deadbeat,majji2012deadbeatController}}

The Information-state contruct is well known (for instance, see \cite[Ch.5]{Bert05}, \cite[Ch.6]{kumar1986stochastic} and the references therein) in the context of optimal control for partially observed LTI systems. However, the salience of the information state with regard to system realization and identification theory was, to the best of our knowledge, not carefully examined. In particular, the information state based identification approach provides a simple yet powerful alternative to traditional realization/ identification approaches as outlined above, in the sense that it allows us to consider the free and forced response together, especially in the context of LTV problems. Perhaps more importantly, assuming that the ultimate purpose of identification is control, the reason for using a finite history of outputs and inputs as the information-state, including the minimum length of such data required, to represent the state of the system, and the optimality of the synthesized feedback control using this finite description, was not considered.\\
\\
\textbf{Contributions:}
The main contributions of this work are (i) the usage of ``information-state" as the state vector, which helps to formulate the state-space realization directly from the ARMA parameters for both time-invariant and time-varying cases without the need for the formation of a Hankel matrix and its SVD. The time-varying case is especially interesting, as the system realization method is exactly the same as the time-invariant case, and doesn't require any coordinate transformations as required in the previous literature~\cite{TV_ERA,TV_OKID};  (ii) a new explanation on why the ARMA model can predict the output from a finite history of inputs and outputs for any linear system based on linear observability, without recourse to a hypothesized deadbeat observer as in \cite{OKID,TV_OKID}; (iii) the approach avoids the need to perform separate experiments for the forced response from zero initial conditions and the free-response from random initial conditions, i.e., there is no need to separate the forced and the free response for identification; and (iv) the implications of using the information-state for optimal output feedback control, where, we show that the solution of a suitably posed optimal control problem using the information-state is indeed optimal for the original output feedback problem. The current manuscript expands on our previously published conference paper \cite{informationstateACC}. In addition to providing a more detailed theoretical development, two important questions are addressed in this paper that were absent in the conference version: (i) the equivalence of the information-state based state-space model and the original state-space system for optimal feedback control (Section~\ref{sec.control_implications}); (ii) identification of the ARMA parameters under noise (Section~\ref{sec.noise}). A numerical experiment to show the equivalence has also been added in Section~\ref{sec.equivalence_exp}. \\
\\
The rest of the paper is structured as follows: Section~\ref{section:prob} introduces the problem; Section~\ref{section:arma_ltv} develops the theory for the information-state approach and provides new insight for the ARMA parameters; Section \ref{sec.control_implications} discusses the implications in optimal feedback control using information-state; Section \ref{sec.noise} discusses the estimation of the ARMA parameters under noise; Owing to the similarities of our work with the OKID/ TV-OKID approaches, Section~\ref{section:okid} is dedicated to showing the relationship and differences between the approaches, in particular, the non-necessity of using an observer in the loop to explain the ARMA model; In Section~\ref{section:results}, we show examples where we apply the approach to a true LTV system and also show its capability of identifying nonlinear systems along a finite trajectory; while Section~\ref{section:Conclusion} draws conclusions about our work.

\section{Problem Formulation}\label{section:prob}
%\textcolor{red}{We need a motivation section here introducing why we want to address this problem, especially the finite time nature of the problem.}
Consider a linear time-varying system given as:
\begin{subequations}\label{eq.system}
    \begin{align}
    x_{t+1} &= A_{t} x_{t} + B_{t} u_{t},  \\
    z_t &= C_{t} x_{t}, 
\end{align}
\end{subequations}

where $x_t \in \mathbb{R}^{n}$ is the state, $u_t \in \mathbb{R}^{r}$ is the input, and $z_t \in \mathbb{R}^{m}$ is the output of the system, defined for $\forall\ t \geq 0$. 
Given the input-output data from such a system with unknown initial conditions, the problem of system identification deals with finding matrices $\mathcal{A}_t$, $\mathcal{B}_t$, $\mathcal{C}_t$ such that the new system:
\begin{align}
    \tilde{x}_{t+1} &= \mathcal{A}_{t} \tilde{x}_{t} + \mathcal{B}_{t} u_{t}, \nonumber \\
    z_t &= \mathcal{C}_{t} \tilde{x}_{t}, \label{eq.id_system}
\end{align}
has the same input-output and transient response as the original system (given in Eq.~\eqref{eq.system}). The dimension of state $\tilde{x}_t$ for the identified system need not be necessarily the same as the dimension of state $x_t$, and thus the dimension of the system matrices $\mathcal{A}_{t}, \mathcal{B}_{t}$ and $\mathcal{C}_{t}$ can also be different from the underlying system matrices. 

\section{Information-State based System Identification: Theory}\label{section:arma_ltv}

In this section, we shall present the theoretical foundation for the information-state based approach for the identification of linear time-varying systems described in Eq.~\eqref{eq.system}. 

\begin{definition}
\textbf{Information-state}. The information-state \cite[Ch.5]{Bert05}, \cite[Ch.6]{kumar1986stochastic} of the system in Eq.~\eqref{eq.system} of order $``q"$ (at time $t$) is defined as 
$$\mathcal{Z}_t = \tr{[\tr{z_{t}}, \tr{z_{t-1}}, \cdots, \tr{z_{t-q+1}}, \tr{u_{t-1}}, \tr{u_{t-2}}, \cdots, \tr{u_{t-q+1}}]},$$ where, $\mathcal{Z}_t \in \mathbb{R}^{mq + r(q-1)}$  
\end{definition}
First, we show the following basic result that is key to the entire development. 
\begin{assumption}\label{assump.observability}
\textit{Observability.}
We assume that the system in Eq.~\eqref{eq.system} is uniformly observable for all time $t$, i.e., the observability matrix:
\begin{equation}
    O^q_{t-1} =
\begin{bmatrix}
 C_{t-1} A_{t-2} \cdots A_{t-q}\\
 C_{t-2} A_{t-3} \cdots A_{t-q}\\
 \vdots \\
 C_{t-q}\\
\end{bmatrix}, \label{eq.obsv}
\end{equation}
is rank $n$ for any $q$ such that $mq \geq n$. 
\end{assumption}

\begin{proposition} \label{prop.1}
Given the linear system in Eq.~\eqref{eq.system}, under Assumption~\ref{assump.observability}, there exists an ARMA model such that:
\begin{align}
    z_t &= \alpha^{(q)}_{t,t-1} z_{t-1} + \alpha^{(q)}_{t,t-2} z_{t-2} +\cdots +   \alpha^{(q)}_{t,t-q} z_{t-q} \nonumber \\ 
    &+\beta^{(q)}_{t,t-1} u_{t-1} + \beta^{(q)}_{t,t-2} u_{t-2} +\cdots + \beta^{(q)}_{t,t-q} u_{t-q}, \label{eq.TV_ARMA}
\end{align}
for any $q\leq t$, and for time $t$ such that $mt \geq n$.
\end{proposition}
\begin{proof}
Note that the outputs of the system can be written as shown in Eq.~\eqref{eq.output_model} (next page). From Eq.~\eqref{eq.output_model}, it follows that (under Assumption~\ref{assump.observability} and for $mt\geq n$): 
$$x_{t-q} = (O_{t-1}^q)^{\dagger} (Z^q_{t-1} - G_{t-1}^q U^q_{t-1}),$$
since $O_{t-1}^q$ has rank $n$ owing to Assumption~\ref{assump.observability}. The symbol $\dagger$ represents the Moore-Penrose inverse. Then, 
\begin{align}
    z_t &= C_{t}A_{t-1}...A_{t-q}x_{t-q}\nonumber\\
    &+[C_t B_{t-1}, \cdots, C_t A_{t-1} \cdots B_{t-q}]U^q_{t-1},\\
    z_t &=C_{t}A_{t-1}...A_{t-q}(O_{t-1}^q)^{\dagger}Z^q_{t-1} + ([C_t B_{t-1}, \cdots,\nonumber \\ 
    C_t& A_{t-1} \cdots B_{t-q}] -C_{t}A_{t-1}...A_{t-q}(O_{t-1}^q)^{\dagger} G_{t-1}^q)U^q_{t-1} . \label{eq.true_arma}
\end{align}
From Eq.~\eqref{eq.true_arma}, it is clear that the ARMA model is given by Eq.~\eqref{eq.TV_ARMA}
where, 
\begin{align}
    [\alpha_{t,t-1}^{(q)} \cdots \alpha_{t,t-q}^{(q)}] &= C_{t}A_{t-1}...A_{t-q}(O_{t-1}^q)^{\dagger}, \nonumber \\
    [\beta_{t,t-1}^{(q)} \cdots \beta_{t,t-q}^{(q)}] &= ([C_t B_{t-1}, \cdots, C_t A_{t-1} \cdots B_{t-q}] \nonumber \\ &-C_{t}A_{t-1}...A_{t-q}(O_{t-1}^q)^{\dagger} G_{t-1}^q). \label{eq.true_ls_sol}
\end{align}
This concludes the proof.
\end{proof}
\begin{table*}[h]
\begin{equation}
    \setcounter{MaxMatrixCols}{20}
    \underbrace{
    \begin{bmatrix}
     z_{t-1} \\
     z_{t-2} \\
     \vdots \\
     z_{t-q}
    \end{bmatrix}}_{Z^q_{t-1}} =
    \underbrace{
    \begin{bmatrix}
     C_{t-1} A_{t-2}\cdots A_{t-q}\\
     C_{t-2} A_{t-3}\cdots A_{t-q}\\
     \vdots\\
     C_{t-q}
    \end{bmatrix}}_{O_{t-1}^q} x_{t-q}+ 
    \underbrace{
    \begin{bmatrix}
  0 & C_{t-1}B_{t-2} & \cdots  & \cdots  & C_{t-1}A_{t-2}...B_{t-q}  \\ 
  0 & 0  & C_{t-2}B_{t-3}  & \cdots  &   C_{t-2}A_{t-3}...B_{t-q} \\ 
  \vdots & \vdots &  \ddots & \cdots  & \vdots  \\
  0&0 &0 & \cdots  & 0  \end{bmatrix}}_{G_{t-1}^q}
  \underbrace{
  \begin{bmatrix}
    u_{t-1}\\
    u_{t-2} \\
    \vdots\\
    u_{t-q}
  \end{bmatrix}}_{U^q_{t-1}}\label{eq.output_model}.
\end{equation}
\end{table*}
Next, we show that the above time-varying ARMA (``TV-ARMA") model is capable of predicting the output $z_t$ for any time $t$ s.t. $mt \geq n$.
\begin{corollary} \label{corollary.1}
For any time $t$, s.t. $mt \geq n$, the TV-ARMA model (Eq.~\eqref{eq.TV_ARMA}) can perfectly predict the output $``z_t"$ of the  linear system (Eq.~\eqref{eq.system}).
\end{corollary}
\begin{proof}
By construction, TV-ARMA (Eq.~\eqref{eq.TV_ARMA}) is capable of predicting $z_t$. Noting that $O_{t-1}^q$ becomes rank ``$n$" for time $t$, for $mt\geq n$, the result follows.
\end{proof}
Next, we show how to find the TV-ARMA model of order $q$ by setting up the problem. 
We may write
\begin{align}
&
\underbrace{
\begin{bmatrix}
  z_t^{(1)} ~  \cdots ~ z_t^{(N)} 
\end{bmatrix}}_{Z_t^{(1:N)}} =
\underbrace{
\begin{bmatrix}
\alpha^{(q)}_{t,t-1}~\cdots~\alpha^{(q)}_{t,t-q}
\end{bmatrix}}_{\bar{\alpha}_t}
\begin{bmatrix}
z_{t-1}^{(1)} & \cdots & z_{t-1}^{(N)} \\
z_{t-2}^{(1)} & \cdots & z_{t-2}^{(N)} \\
\vdots & \ddots & \vdots \\
z_{t-q}^{(1)} & \cdots & z_{t-q}^{(N)} \\
\end{bmatrix}
 \nonumber\\ &~~~~~~+ 
 \begin{bmatrix}
\underbrace{\beta^{(q)}_{t,t-1}~\cdots~\beta^{(q)}_{t,t-q}}_{\bar{\beta}_t}
\end{bmatrix}
\begin{bmatrix}
u^{(1)}_{t-1} & \cdots & u^{(N)}_{t-1}\\
u^{(1)}_{t-2} & \cdots & u^{(N)}_{t-2}\\
\vdots & \ddots & \vdots \\
u^{(1)}_{t-q} & \cdots & u^{(N)}_{t-q}\\
\end{bmatrix}
,\label{eq.ls_eq}
\end{align}
\noindent where ${z^{(i)}_{t}, u^{(i)}_{t}}$ represent the observation and inputs for the $``i^{\text{th}}"$ rollout/experiment. $N$ is the total number of rollouts and is chosen such that $N > mq + rq.$\\
The order of the state-space model can then be found using the following approach. Let the concatenated data matrix be represented as :
\begin{align}
    \mathcal{X}^q_t = 
    \begin{bmatrix}
    z_{t-1}^{(1)} & \cdots & z_{t-1}^{(N)} \\
    \vdots & \ddots & \vdots \\
    z_{t-q}^{(1)} & \cdots & z_{t-q}^{(N)}\\
     u^{(1)}_{t-1} & \cdots & u^{(N)}_{t-1} \\
     \vdots & \ddots & \vdots \\
     u^{(1)}_{t-q} & \cdots & u^{(N)}_{t-q} 
    \end{bmatrix} \in \mathbb{R}^{mq+rq \times N}. \label{eq.data_matrix}
\end{align}
% \textcolor{blue}{Can we show that the rank of $\mathcal{X}^q_t $ is maximum of $``n+ (q+1)r"$, \textbf{with $n$ being the order of the underlying system}?}\\

% \textcolor{blue}{The rank of data matrix $\mathcal{X}^q_t$ keeps increasing with the increase in the order of ARMA parameters $q$ and reaches a maximum value of $``n+ (q+1)r"$, with $n$ being the order of the underlying system. Further increase in order of ARMA model, i.e., $mq^* > n$, results in a row rank deficient data matrix $\mathcal{X}^q_t$.}
The data matrix $\mathcal{X}^q_t$ is a full row rank matrix for $mq < n$, with $n$ being the order of the underlying system. Therefore, the minimum order of the ARMA parameters $q$ is obtained by increasing the order $q$ till the data matrix $\mathcal{X}^q_t$ becomes row rank deficient. 
In particular, we can do an SVD of $\mathcal{X}^q_t$ for some suitably large $q$. The rank of the $\mathcal{X}^q_t$ is going to be $``n+ rq"$, and thus, the minimal order ARMA model can be found by setting $q^*$ s.t. $mq^* \geq n$. In particular, for single-output system $(m=1)$, the minimal order $q^* = n$. However, in general owing to the fact that $m$ and $n$ need not be integer multiples, i.e., $n \neq mq$ for some integer $q$, the order $q$ of the ARMA model is going to be larger than the state order $n$.  The above development may now be summarized as follows.  
\begin{proposition}
Given a data matrix $\mathcal{X}^q_t$ of order $q$ as in Eq.~\eqref{eq.data_matrix}, with $mq > n$, the minimal order ARMA model for system (Eq.~\eqref{eq.system}) is given by choosing $q^*$ s.t. $mq^* \geq n$, where $(n+rq)$ is the rank of $\mathcal{X}^q_t$.
\end{proposition}

\subsection{Solving for the ARMA parameters:} \label{sec.sol_arma}
Consider the linear problem:
\begin{equation}
    Z^{(1:N)}_t = \begin{bmatrix}
    \bar{\alpha}_t &  \bar{\beta}_t
    \end{bmatrix}\mathcal{X}^q_t 
    .\label{eq.ls_eq_abbrev}
\end{equation}
where $Z^{(1:N)}_t \in \mathbb{R}^{m \times N}$ indicate the outputs from $N$ rollouts, and $\mathcal{X}^q_t \in \mathbb{R}^{mq+rq \times N}$ is the data matrix containing the past outputs and inputs as shown in Eq.~\eqref{eq.data_matrix}. The ARMA parameters can be solved by minimizing the least-squares error.  \\
For $q$ s.t. $mq \geq n$, Eq.~\eqref{eq.ls_eq_abbrev} does not have a unique solution, owing to the row rank-deficiency of $\mathcal{X}^q_t$. The least-squares solution is the Moore-Penrose inverse of Eq.~\eqref{eq.ls_eq_abbrev}, assuming $N > (m+r)q$. Let $\mathcal{X}^q_t = U 
\Sigma  \tr{V}$ where $\Sigma = \begin{bmatrix}
\Sigma_{n + rq} & \\
& \mathbf{0} \\
\end{bmatrix} $, $\Sigma_{n + rq} = \begin{bmatrix}
\sigma_1 &  &  0\\
& \ddots & \\
0 &  &\sigma_{n + qr}
\end{bmatrix}$, and $V = [V_1\ V_2]$, $U = [U_1\ U_2]$, where $(U_1,\ V_1)$ are the singular vectors corresponding to the non-zero singular values $\Sigma_{n + rq}$ and $(U_2,\ V_2)$ are the singular vectors corresponding to the zero-singular values, i.e., the singular vectors spanning the nullspace of $\mathcal{X}^q_t$. The least-squares (LS) solution is: 
\begin{equation}
    \begin{bmatrix}
    \bar{\alpha}^*_t  & \bar{\beta}^*_t
    \end{bmatrix} = \sum_{i=1}^{n + rq} \sigma^{-1}_i (Z^{(1:N)}_t \mathrm{v}_i)\tr{\mathrm{u}_i} , \label{eq.ls_sol}
\end{equation}
where $\mathrm{u}_i$ and $\mathrm{v}_i$ are the singular vectors corresponding to $\sigma_i$.\\
\indent A natural question that arises is: \textit{What is the relationship between the ARMA model of Eq.~\eqref{eq.true_ls_sol} and the least-squares problem given by Eq.~\eqref{eq.ls_eq_abbrev}}? First, we show that the ARMA solution (Eq.~\eqref{eq.true_ls_sol}) satisfies the LS-problem Eq.~\eqref{eq.ls_eq_abbrev}.
\begin{proposition}
The ARMA solution given by Eq.~\eqref{eq.true_ls_sol} is a solution of the LS-equation~\eqref{eq.ls_eq_abbrev}.
\end{proposition}
\begin{proof}
By construction, $z_t = \alpha^{\#}_{t,t-1} z_{t-1} + \cdots  + \beta^{\#}_{t,t-1} u_{t-1} + \cdots + \beta^{\#}_{t,t-q} u_{t-q},$ where $(\alpha^{\#}_{t,t-k},  \beta^{\#}_{t,t-k})$ are the solution~\eqref{eq.true_ls_sol}, and it holds for all ${z^{(i)}_t}$
\end{proof}
Thus, given the system~\eqref{eq.system}, after time $t$, s.t. $mt \geq n$, the output $z_t$ can be predicted by both the LS-solution~\eqref{eq.ls_sol} as well as the ARMA-solution~\eqref{eq.true_ls_sol}. This is summarized in the following corollary.
\begin{corollary} \label{corollary_ls}
Let the solution of Eq.~\eqref{eq.true_ls_sol} be denoted by $(\bar{\alpha}^{\#}_t, \bar{\beta}^{\#}_t)$. Then for any time $t$ s.t. $mt \geq n$, 
\begin{align}
    z_t &= \alpha^{\#}_{t,t-1} z_{t-1} + \cdots + \alpha^{\#}_{t,t-q} z_{t-q} \nonumber\\ 
    &  + \beta^{\#}_{t,t-1} u_{t-1} + \cdots + \beta^{\#}_{t,t-q} u_{t-q} \\ 
    & = \alpha^{*}_{t,t-1} z_{t-1} + \cdots + \alpha^{*}_{t,t-q} z_{t-q} \nonumber\\ 
        &  + \beta^{*}_{t,t-1} u_{t-1} + \cdots + \beta^{*}_{t,t-q} u_{t-q}. \label{eq.ls_sol_arma} 
\end{align}
An immediate consequence of the above result is that the LS-ARMA model~\eqref{eq.ls_sol} also predicts the response of the linear system~\eqref{eq.system}.
\end{corollary}
\begin{figure}[t]
    \centering
    \includegraphics[width=0.8\linewidth]{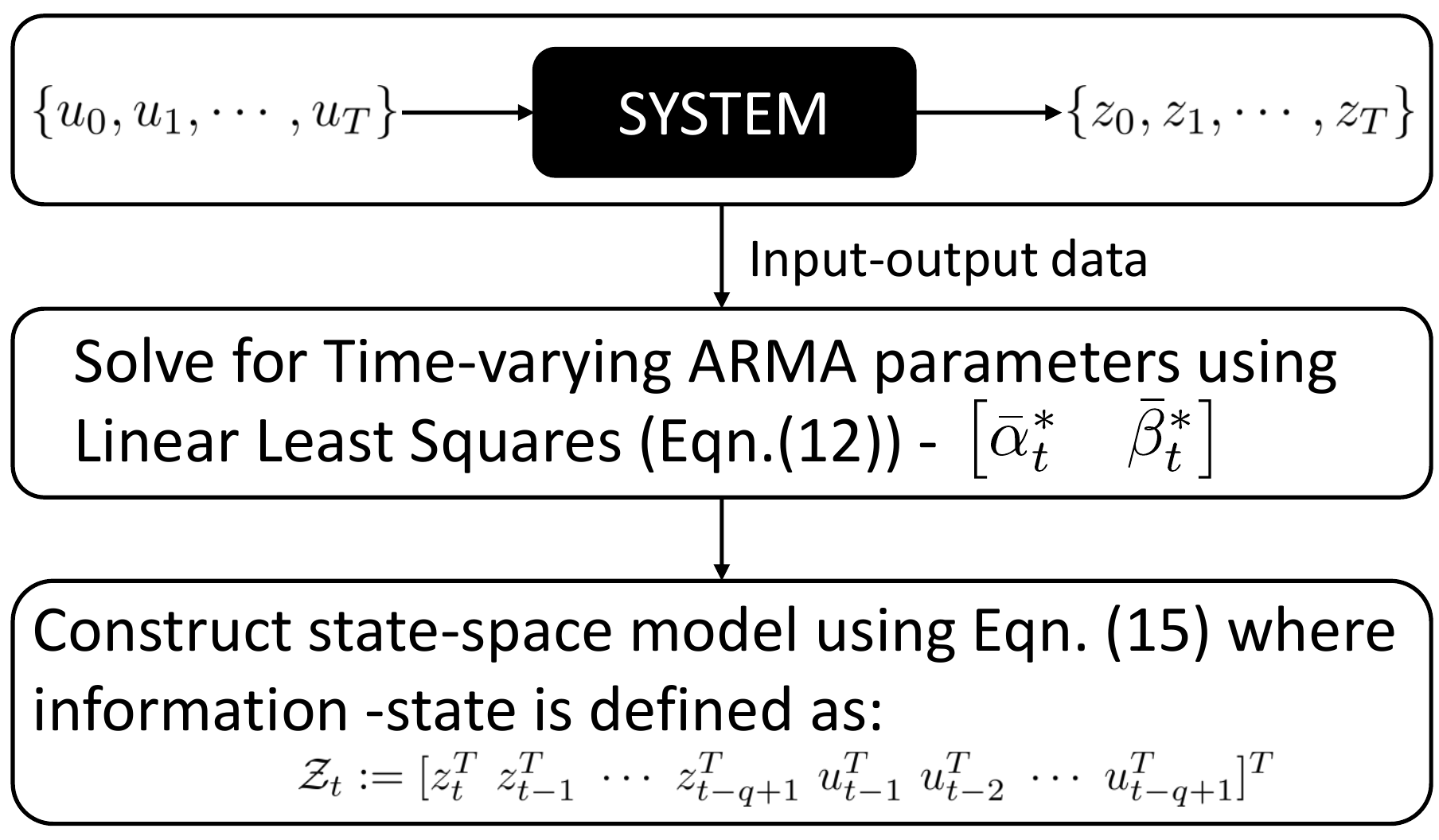}
    \caption{Steps involved in system identification using Information-state approach.}
    \label{fig:intro_image}
\end{figure}

We note here that the LS-solution~\eqref{eq.ls_sol} and the ``fundamental solution"~\eqref{eq.true_ls_sol} need not be the same; nonetheless, both can predict the initial + forced response of the linear system~\eqref{eq.system} after time $t$ s.t. $mt \geq n$, under the observability assumption~\ref{assump.observability}.\\
\begin{comment}
\mxx{Now, we establish the equivalence between ARMA models with different $q$ values.
\begin{corollary} \label{corol.arma length equivalence}
    Let $q' \geq q$, where $q$ satisfies assumption~\ref{assump.observability}. Denote the ARMA model, given by eq.~\eqref{eq.ls_sol}, corresponding to length $q'$ and $q$ as $(\bar{\alpha}_t^{(q')}, \bar{\beta}_t^{(q')})$ and $(\bar{\alpha}_t^{(q)}, \bar{\beta}_t^{(q)})$, respectively. Then, for any time $t$, such that $t\geq q'$,
    \begin{align*}
         z_t &= \alpha^{(q')}_{t,t-1} z_{t-1} + \cdots + \alpha^{(q')}_{t,t-q'} z_{t-q'} \\ 
    &  + \beta^{(q')}_{t,t-1} u_{t-1} + \cdots + \beta^{(q')}_{t,t-q'} u_{t-q'} \\ 
    & = \alpha^{(q)}_{t,t-1} z_{t-1} + \cdots + \alpha^{(q)}_{t,t-q} z_{t-q} \\ 
        &  + \beta^{(q)}_{t,t-1} u_{t-1} + \cdots + \beta^{(q)}_{t,t-q} u_{t-q}. 
    \end{align*}
\end{corollary}
\begin{proof}
    From corollary \ref{corollary_ls}, we know that the ARMA model given by eq.~\eqref{eq.ls_sol}, with any length $q$ that satisfies $mq \geq n$, predicts the output exactly. Hence, the output from both the ARMA models with length $q'$ and $q$ has to be equal.  
\end{proof}
}
\end{comment}
The information-state based state-space model can then be written as in Eq.~\eqref{eq.info_state_model}(next page).
\begin{table*}[h]
\begin{subequations}\label{eq.info_state_model}
\begin{align}
    \setcounter{MaxMatrixCols}{20}
    \underbrace{
    \begin{bmatrix} 
    z_{t} \\ z_{t-1} \\ z_{t-2} \\ \vdots \\ z_{t-q+1} \\ \hline u_{t-1} \\u_{t-2} \\ u_{t-3} \\ \vdots \\ u_{t-q+1} 
    \end{bmatrix}}_{\mathcal{Z}_t} &= 
    \underbrace{\begin{bmatrix} 
    \alpha_{t,t-1} & \alpha_{t,t-2} & \cdots  &  \alpha_{t,t-q+1}  & \alpha_{t,t-q} & \vline &  \beta_{t,t-2} & \beta_{t,t-3} & \cdots  &  \beta_{t,t-q+1}  
    & \beta_{t,t-q}    \\  
    I_{m \times m} & 0 & \cdots & 0 & 0 & \vline & 0 & 0 & \cdots & 0 
    & 0    \\  
    0 & I & \cdots & 0 & 0 & \vline & 0 & 0 & \cdots & 0 
    & 0    \\  
    \vdots & & \ddots &  & \vdots & \vline & \vdots &  & \ddots & \vdots
    & 0    \\
    0 & 0 & \cdots & I & 0 & \vline & 0 & 0 & \cdots & 0 
    & 0    \\
    \hline
    0 & 0 & \cdots & 0 & 0 & \vline & 0 & 0 & \cdots & 0 
    & 0    \\
    0 & 0 & \cdots & 0 & 0 & \vline & I_{r\times r} & 0 & \cdots & 0 
    & 0     \\  
     0 & 0 & \cdots & 0 & 0 & \vline & 0 & I & \cdots & 0 
    & 0     \\  
    \vdots & & \ddots &  & \vdots & \vline & \vdots &  & \ddots & & \vdots    \\
    0 & 0 & \cdots & 0 & 0 & \vline & 0 & 0 & \cdots & I
    & 0
    \end{bmatrix} }_{\mathcal{A}_{t-1}}
    \underbrace{\begin{bmatrix} z_{t-1} \\ z_{t-2} \\ \vdots \\ z_{t-q+1}  \\ z_{t-q} \\ \hline  u_{t-2} \\ u_{t-3} \\   \vdots \\u_{t-q+1} \\ u_{t-q} \end{bmatrix}}_{\mathcal{Z}_{t-1}} 
    + \underbrace{\begin{bmatrix} \beta_{t,t-1} \\ 0 \\ \vdots \\ 0 \\ 0 \\ \hline I \\ 0  \\ \vdots \\  0 \\0 \end{bmatrix}}_{\mathcal{B}_{t-1}}  u_{t-1} \\
    z_t &= \underbrace{\begin{bmatrix}
    I_{m\times m} & 0 &\cdots & 0
    \end{bmatrix}}_{\mathcal{C}_t} \mathcal{Z}_{t}  \label{eq.info_state_output}
\end{align} 
\end{subequations}
\end{table*}
It can now be seen that the information-state $\mathcal{Z}_{t}$ is indeed a state of the system~\eqref{eq.system}. In particular, the following result holds:
\begin{proposition} \label{prop.model_equivalence}
Under assumption~\ref{assump.observability}, and for any time $t$, s.t. $mt \geq n$, the output $z'_t$ of the information-state model in Eq.~\eqref{eq.info_state_model} is equal to the output of the linear state-space model~\eqref{eq.system}, given the same inputs, and regardless of the ``unknown" initial condition $x_0$.
\end{proposition}
\begin{proof}
By corollary~\ref{corollary_ls}, we know that the output of the linear system~\eqref{eq.system} after time $t$ s.t. $mt\geq n$ may be written as: $z_t = \alpha^{*}_{t,t-1} z_{t-1} + \cdots + \alpha^{*}_{t,t-q} z_{t-q} + \beta^{*}_{t,t-1} u_{t-1} + \cdots + \beta^{*}_{t,t-q} u_{t-q}. $ Further $\mathcal{Z}_t$ is a state of the system owing to its definition, and thus, noting the definition of the output $z'_t$ of the information-state system~\eqref{eq.info_state_output}, the result immediately follows.
\end{proof}
\begin{comment}
    \mxx{
It is straightforward to see that one can construct the information-state model for different values of $q$ for the same system as long as $mq \geq n$, and all such models will have the same output, owing to corollary \ref{corol.arma length equivalence}. This is summarized in the following result. 

\begin{corollary} \label{corol.info_state model equivalence}
    Let $q' \geq q$, where $q$ satisfies assumption~\ref{assump.observability}. For any time $t$, s.t. $t \geq q'$, the information-state models constructed with length $q'$ and $q$ predict the output of the linear state-space model~\eqref{eq.system}, given the same inputs, and regardless of the ``unknown" initial condition $x_0$.
\end{corollary}

}
\end{comment}

\subsection*{Linear Time Invariant (LTI) Case:} A further simplification of the information-state model is possible if the linear system~\eqref{eq.system} is LTI. Since the system is LTI, Eq.~\eqref{eq.TV_ARMA} becomes $z_t = \alpha_1 z_{t-1} + \alpha_2 z_{t-2} + \cdots + \alpha_q z_{t-q} + \beta_1 u_{t-1} + \cdots + \beta_q u_{t-q}$ regardless of time $t$. Let us define:
\begin{align}
    X^{(1)}_t := z_t = \alpha_1 X^{(1)}_{t-1} + X^{(2)}_{t-1} + \beta_1 u_{t-1} ,\label{eq.observer_x1}
\end{align}
where,
\begin{align}
    X^{(2)}_{t-1} &:= \alpha_2 z_{t-2} + \cdots + \alpha_q z_{t-q} + \beta_2 u_{t-2} + \cdots + \beta_q u_{t-q}, \nonumber\\
    X^{(2)}_{t} &:= \alpha_2 z_{t-1} + \cdots + \alpha_q z_{t-q + 1}  \nonumber\\
    &+\beta_2 u_{t-1} + \cdots + \beta_q u_{t-q +1}, \label{eq.observer_x2}.
\end{align}
Let us also define: 
\begin{align}
    X^{(3)}_{t-1} &:= \alpha_3 z_{t-2} + \cdots + \alpha_q z_{t-q+1} + \nonumber\\
    &\beta_3 u_{t-2} + \cdots + \beta_q u_{t-q+1}.\nonumber
\end{align}
Now, we can write $X^{(2)}_{t}$ as:  
\begin{align}
     X^{(2)}_{t} = \alpha_2 X^{(1)}_{t-1} + X^{(3)}_{t-1} + \beta_2 u_{t-1}. \label{eq.observer_x2_ss}
\end{align}
Continuing in this way till $X^{(q)}_{t}$, we obtain:
\begin{align}
     X^{(q)}_{t} = \alpha_q X^{(1)}_{t-1} + \beta_q u_{t-1} \label{eq.observer_xq}.
\end{align}
Putting Eq.~\eqref{eq.observer_x1} -~\eqref{eq.observer_xq} together, we obtain the information-state dynamics in the ``observer canonical form":\\
\begin{align}
   \underbrace{\begin{bmatrix}
    X^{(1)}_{t} \\ \vdots \\ X^{(q)}_{t}
    \end{bmatrix}}_{\bar{X}_t} &=
    \begin{bmatrix}
    \alpha_1 & I & 0 & \cdots & 0\\
    \alpha_2 & 0 & I & \cdots & 0\\
    \vdots & &\ddots & & \vdots\\
    \alpha_q & 0 & 0 & \cdots & 0\\
    \end{bmatrix}
    \begin{bmatrix}
    X^{(1)}_{t-1} \\ \vdots \\ X^{(q)}_{t-1}
    \end{bmatrix}+ 
    \begin{bmatrix}
    \beta_1 \\ \beta_2 \\ \vdots \\ \beta_q
    \end{bmatrix} u_{t-1}, \\
    z_t &= 
    \begin{bmatrix}
    I & 0 & \cdots & 0
    \end{bmatrix}
    \bar{X}_t .
\end{align}

 We note that the above formulation to create the information state and the linear state space model in observer canonical form is restricted to the LTI case, as the simplification in equations~(\ref{eq.observer_x1}-\ref{eq.observer_xq}) are feasible only in the LTI case: since, in the LTV case, the coefficients $\alpha, \beta$ are also time dependent, we cannot obtain Eq.~\eqref{eq.observer_x2_ss} from Eq.~\eqref{eq.observer_x2}. Therefore, the information-state in the LTV case needs to include the past inputs $\{u_{t-1}, \cdots, u_{t-q}\}$ as part of the information-state.

\section{Implication on feedback control using the information-state} \label{sec.control_implications}
One of the primary reasons for identifying a system is for the purposes of control. In this section, we theoretically show that designing an optimal feedback control using the information-state model leads to no loss in optimality when compared to the full-state system after an initial transient. First, we show that the true state and the information-state are related by a linear transformation in the following lemma.

\begin{lemma}\label{lemma.lineartransformation}
The true state $x_t$ is a linear transformation of the information-state $\mathcal{Z}_t$, i.e. $x_t = T_t \mathcal{Z}_t$. 
\end{lemma}
\begin{proof}
    From~\eqref{eq.output_model}, we know $Z_{t-1} = O_{t-1}x_{t-q} + G_{t-1}U_{t-1}$ (we leave out the superscript $q$ for readability). Then,
    \begin{align*}
        x_{t-q} &= (O_{t-1})^{\dagger} (Z_{t-1} - G_{t-1} U_{t-1}) \\
        x_{t-q} &= (O_{t-1})^{\dagger} [I_{mq\times mq} - G_{t-1}]
        \underbrace{
        \begin{bmatrix}
                Z_{t-1} \\
                U_{t-1}
        \end{bmatrix}}_{\mathcal{Z}_{t-1}} \\
        x_{t-q} &= \tilde{T}_{t-1}\mathcal{Z}_{t-1}.
    \end{align*}
    The state at time $t-1$ can be written as 
    \begin{align*}
        x_{t-1} &= A_{t-2} \cdots A_{t-q} x_{t-q} + \nonumber \\
        &[B_{t-2},\ A_{t-2}B_{t-3}, \cdots ,A_{t-2}\cdots B_{t-q}] U_{t-1}, \\
        x_{t-1} &=  A_{t-2} \cdots A_{t-q} \tilde{T}_{t-1}\mathcal{Z}_{t-1} + \nonumber\\
        & [0_{n\times mq},\ B_{t-2},\ \cdots ,\ A_{t-2}\cdots B_{t-q}] \mathcal{Z}_{t-1}, \nonumber\\
        x_{t-1} &= T_{t-1} \mathcal{Z}_{t-1} .
    \end{align*}
\end{proof}
\begin{comment}
\begin{corollary} \label{corol.unique transformation}
    For a given value of the parameter $q$ of the ARMA model, there exists a unique map from $\mathcal{Z}_{t}$ to $x_t$ and vice-versa.
\end{corollary}
\begin{proof}
    The unique map from $\mathcal{Z}_{t}$ to $x_t$ is due to Lemma~\ref{lemma.lineartransformation}. 
    The map from $x_t$ to $\mathcal{Z}_{t}$
    is also unique from the definition of $\mathcal{Z}_{t}$, which is a finite collection of the past outputs and inputs that led to state $x_t$.
\end{proof}    
\end{comment}
To show the optimality of a feedback law using the information-state, we consider the following finite horizon optimal control problem,
\begin{subequations}\label{eq.OCP}
\begin{align}
    J^* &= \min_{u_0, \cdots, u_{N-1}} \sum_{t=0}^{H-1} c_t(z_t, u_t) + c_N(z_N), \\ 
    \text{subject to:} ~~ x_t &= A_{t-1}x_{t-1} + B_{t-1} u_{t-1}, \nonumber \\
    z_t &= C_t x_t,\label{eq.OCP_dynamics} \\
    x_0 & \text{ is unknown,} \nonumber 
\end{align}
\end{subequations}
where, $c_t(\cdot, \cdot)$ and $c_N(\cdot)$ are given convex cost functions and and $H$ is the horizon. The solution to the above optimal control problem cannot be calculated since the initial state is unknown. \\
Let $\bar{q}$ be the minimal $q$ that satisfies $mq\geq n$.
Define 
\begin{align*}
    J_1 &= \sum_{t=0}^{\bar{q}-2} c(z_t,u_t), \\
    J_x &= \sum_{t=\bar{q}-1}^{H-1} c(z_t,u_t) + c_N(z_N). 
\end{align*}
Eq.~\eqref{eq.OCP} can be decoupled and written as 
\begin{align} \label{eq.decoupled OCP}
    J^* \leq \min_{u_0, \cdots, u_{\bar{q}-2}} J_1 + \min_{u_{\bar{q}-1}, \cdots, u_{N-1}} J_x 
\end{align}
subject to the constraints in Eq.~\eqref{eq.OCP_dynamics}. \\
 We assume that the first $\bar{q} - 1$ control inputs and resulting observations are specified to us which implies that the the state $x_{\bar{q}-1}$ is known given the observability assumption. Then, we reformulate the original partially-observed control problem as the following fully observed problem starting at the initial state $x_{\bar{q}-1}$:
 \begin{subequations}\label{eq:OptProbX}
     \begin{align}
    J_x^* &= \min_{u_{\bar{q}-1}, \cdots, u_{N-1}} \sum_{t=\bar{q}-1}^{H-1} c(z_t,u_t) + c_N(z_N). \label{eq.cost_J_x}\\
    \text{s.t.}~ x_t &= A_{t-1}x_{t-1} + B_{t-1} u_{t-1}, ~\text{given}~ x_{\bar{q}-1}\nonumber \\
    z_t &= C_t x_t,\label{eq:OptProbX_dynamics}
\end{align}
 \end{subequations}

The above problem is well-posed and since the full state is known (due to lack of uncertainty), there exists an optimal sequence of control inputs $u_t^* = \pi_x^t(x_t)$, where $\pi_x^t(\cdot)$ is the underlying optimal feedback policy at time $t$. \\
\begin{remark}\label{remark.III.1}
    Note that typically for a partially observed problem, one would use an observer to estimate the state, which would then be used to specify the control action assuming the estimated state to be the true state (certainty equivalence). Owing to the observability assumption \ref{assump.observability}, given the first $\bar{q}$ inputs and outputs, one can, in principle, exactly reconstruct the initial state, and thus predict the state evolution exactly after $\bar{q}$ steps, thereby mapping back to the fully observed problem. So, we shall assume that the first $\bar{q}$ inputs are specified in the partially observed control formulation we discuss later in this section.
\end{remark}
We now pose the following ``fully observed" optimal control problem in terms of the Information-State assuming that the same initial $\bar{q}$ inputs are applied as above:
% \vspace{-1mm}
\begin{subequations}\label{eq:OptProbZ}
\begin{align}
    J_{\mathcal{Z}}^* &= \min_{u_{\bar{q}-1}, \cdots, u_{N-1}}  \sum^{H-1}_{t=\bar{q}-1} c(z_t, u_t) + c_T(z_T), \\
    \text{s.t.} ~~ 
    \mathcal{Z}_{t} &= \mathcal{A}_{t-1}\mathcal{Z}_{t-1} + \mathcal{B}_{t-1}u_{t-1},\text{given}\   \mathcal{Z}_{\bar{q}-1} \label{eq:OptProbZ_dynamics}
\end{align}
\end{subequations}
where, the information-state is of order $\bar{q}$. The equivalence of the solution to the two problems is established in the following result.
\begin{comment}
From Proposition~\ref{prop.model_equivalence}, we know that the information-state model has the same input-output response as the true dynamics in Eq.~\eqref{eq.OCP_dynamics_1}. We also know that the initial $\mathcal{Z}_{\bar{q}-1}$ is unique mapping of $x_{\bar{q}-1}$. We show the following result.
\end{comment}
\begin{theorem}\label{theorem.optimal_fb}
    Given the initial information-state $\mathcal{Z}_{\bar{q}-1}$ or equivalently the state $x_{\bar{q}-1}$, the solution to the optimal control problem (\ref{eq:OptProbX}) is identical to the solution to the optimal information state control problem (\ref{eq:OptProbZ}).
\end{theorem}
\begin{proof}
Consider the two problems in eq. \eqref{eq:OptProbX} and \eqref{eq:OptProbZ}.
To show the equivalence between the two problems, one needs to show: (i) the initial conditions are equivalent; (ii) the input-output response of the state-space system and information-state system are equal; (iii) the optimal feedback policies of both the systems generate the same control inputs.\\
Using lemma~\ref{lemma.lineartransformation}, the initial condition for the systems are related by $x_{\bar{q}-1} = T_{\bar{q}-1}\mathcal{Z}_{\bar{q}-1}$, where the information-state contains the inputs and outputs of the system at $t = 0, \cdots, \bar{q}-2 $. Proposition~\ref{prop.model_equivalence} shows the state-space model Eq.~\eqref{eq:OptProbX_dynamics} and information-state model Eq.~\eqref{eq:OptProbZ_dynamics} have the same input-output response. \\
\indent To show that the feedback policies generate the same control inputs,  let the optimal control policy mapping the state to the control at time $t$ for the state-space and information-space system be $\pi^t_x(\cdot)$, and $\pi_{\mathcal{Z}}^t(\cdot)$, respectively.     
The policy for the true full-state feedback $\pi_x^t(\cdot)$ can be applied to the information-state system, by using the transformation from lemma~\ref{lemma.lineartransformation}, $x_t =T_t \mathcal{Z}_{t}$, which gives the control input to be  $u_t = \pi_x^t(T_t \mathcal{Z}_{t})$. Then, 
because of optimality of the policy $\pi_{\mathcal{Z}}(\cdot)$ for the information-state system, the input-output equivalence of the systems in Eqs.~\eqref{eq:OptProbX_dynamics} and~\eqref{eq:OptProbZ_dynamics}, and the same cost functions, we get:
 \begin{align} \label{eq.prop5_1}
    J_{\mathcal{Z}}^{\pi_{\mathcal{Z}}} \leq J_x^{\pi_x}.
\end{align}
Similarly, the optimal information-state policy  $\pi_{\mathcal{Z}}(\cdot)$, can be applied to the true state-space model by taking the past inputs and outputs as the argument to the policy. Because of optimality of $\pi_x$ for the state space system, we get 
\begin{align}\label{eq.prop5_2}
    J_x^{\pi_x} \leq J_{\mathcal{Z}}^{\pi_{\mathcal{Z}}} 
\end{align}
From Eq.~\eqref{eq.prop5_1} and Eq.~\eqref{eq.prop5_2}, $J_x^{\pi_x} =J_{\mathcal{Z}}^{\pi_{\mathcal{Z}}} $. Hence, the optimal feedback on the information-state is identical to the optimal control for the underlying state-space system. 
\end{proof}
\begin{remark} \label{remark.2}
Please notice that the above two defined problems (Eq.~\eqref{eq:OptProbX} and Eq.~\eqref{eq:OptProbZ}) specify the initial state at $t = \bar{q}-1$ by assuming that the initial $\bar{q}-1$ control inputs $\{u_0, u_1, \cdots, u_{\bar{q}-2}\}$ are specified. Intuitively, one can see that these inputs are inherently ambiguous as they are needed to get enough information to reconstruct the initial state or form the initial information state and should not be included in the optimization. Thus, it is advisable to take small control perturbations to keep the system near the initial state till the state can be reconstructed at $t = \bar{q}-1$. Alternatively, one could use the ``best" observer to estimate the state of the system and use it for control. Owing to observability, no observer would be able to reconstruct the state exactly before $t = \bar{q} - 1$. Hence, the best observer based system will still generate optimal control inputs only from $t = \bar{q} - 1$, which the information-state based control also guarantees. Nonetheless, given that $\bar{q} \ll H$, we can expect that this initial transient should not affect the total cost significantly.
\end{remark}

As discussed in section~\ref{section:arma_ltv}, there might be multiple solutions for the ARMA parameters (Eq.~\eqref{eq.ls_eq_abbrev}), which will lead to multiple realizations of the information-state model. In the following result, we show that the realization does not matter when it comes to designing for the optimal control. 
\begin{proposition}
    For a specific value $q$, and given they start from the same initial condition, all realizations of the information-state model arising from different ARMA parameters that satisfy Eq. \eqref{eq.ls_eq_abbrev}, generate the same optimal feedback control. 
\end{proposition}
\begin{proof}
Suppose there are two realizations of the information-state model, arising from using the ARMA model from Eqs.~\eqref{eq.true_ls_sol} and \eqref{eq.ls_sol}. Consider the following problems:
\begin{align}
    \mathcal{J}^{(1), \bar{\pi}^*} &= \min_{\bar{\pi}(\mathcal{Z}_{t})} \sum_{t=q-1}^{H-1} c(z_t,u_t) + c_N(z_N), \nonumber \\
    \text{subject to:} ~~\mathcal{Z}_{t} &= \mathcal{A}_{t-1}^{(1)}\mathcal{Z}_{t-1} + \mathcal{B}_{t-1}^{(1)} u_{t-1}, \text{given}\  \mathcal{Z}_{q-1} \nonumber \\
    u_t &= \bar{\pi}_{t}(\mathcal{Z}_{t}),  \nonumber
\end{align}
and, 
\begin{align}
    \mathcal{J}^{(2), \bar{\mu}^*} &= \min_{\bar{\mu}(\mathcal{Z}_{t})} \sum_{t=q-1}^{H-1} c(z_t,u_t) + c_N(z_N), \nonumber \\
    \text{subject to:} ~~\mathcal{Z}_{t} &= \mathcal{A}_{t-1}^{(2)}\mathcal{Z}_{t-1} + \mathcal{B}_{t-1}^{(2)} u_{t-1},  \mathcal{Z}_{q-1} \nonumber \\
    u_t &= \bar{\mu}_{t}(\mathcal{Z}_{t}). \nonumber
\end{align}

From Corollary 2, we know that different ARMA models predict the same output. Hence, the different information-state models also have the same input-output response, and also start from the same initial condition. To show they generate the same optimal control, consider applying the optimal policy of the first system $\bar{\pi}^*$ to the second system. Then we know, $\mathcal{J}^{(2), \bar{\mu}^*} \leq \mathcal{J}^{(2), \bar{\pi}^*} $, owing to optimality of $\bar{\mu}^*$ to the second system. But since both the models have the same input-ouput response and the cost functions, $\mathcal{J}^{(2), \bar{\pi}^*} = \mathcal{J}^{(1), \bar{\pi}^*}$. Hence,
\begin{align}
    \mathcal{J}^{(2), \bar{\mu}^*} \leq \mathcal{J}^{(1), \bar{\pi}^*}. 
\end{align}
Similarly, by applying the policy $\bar{\mu}^*$ to the first system, we can show that 
\begin{align}
    \mathcal{J}^{(1), \bar{\pi}^*} &\leq \mathcal{J}^{(2), \bar{\mu}^*} 
\end{align}
From the above two sets of equations, we can conclude that, 
\begin{align}
    \mathcal{J}^{(1), \bar{\pi}^*} &= \mathcal{J}^{(2), \bar{\mu}^*}
\end{align}
Since both problems have unique solutions, owing to the convexity of the problems, they should be the same control inputs for both realizations.
\end{proof}
We now find the optimal length $q$ for the information-state. 
\begin{proposition}
    The optimal value for the information-state length $q$, to solve the optimal control problem in Eq.~\eqref{eq.OCP}, is the minimum integer $q$ that satisfies $m q \geq n $.
\end{proposition}
\begin{proof}
Let $\bar{q}$ be the minimum value of $q$ that satisfies $mq \geq n$. Define $q' = \bar{q} + l$, where $l$ is any positive integer.
Consider Eq. \eqref{eq.decoupled OCP} which is the decoupled optimal control problem in Eq.~\eqref{eq.OCP}. 
Using Theorem~\ref{theorem.optimal_fb}, Eq. \eqref{eq.decoupled OCP} can be rewritten as 
\begin{align}
    J^* \leq \min_{u_0, \cdots, u_{\bar{q}-2}} \sum_{t=0}^{\bar{q}-2} c(z_t,u_t) + J^*_{\mathcal{Z}^{\bar{q}}},
\end{align}
where, 
\begin{align}
    J^*_{\mathcal{Z}^{\bar{q}}} &= \min_{u_{\bar{q}-1}, \cdots, u_{N-1}} \sum_{t=\bar{q}-1}^{H-1} c(z_t,u_t) + c_N(z_N), \nonumber \\
   \text{subject to:} ~~ \mathcal{Z}^{\bar{q}}_{t} &= \mathcal{A}^{\bar{q}}_{t-1}\mathcal{Z}^{\bar{q}}_{t-1} + \mathcal{B}^{\bar{q}}_{t-1}u_{t-1},  \text{given}\ \mathcal{Z}^{\bar{q}}_{\bar{q}-1}. \nonumber
\end{align}
As discussed in Remark \ref{remark.2}, the first few control inputs - $u_0, \cdots, u_{\bar{q}-2}$ - till the information-state is constructed are specified and are suboptimal. Define the total cost using these specified control inputs as
\begin{align}\label{eq.OCP q_bar}
    J_{\bar{q}} = \sum_{t=0}^{\bar{q}-2} c(z_t,u_t) + J^*_{\mathcal{Z}^{\bar{q}}}.
\end{align}
Similarly, the equivalent cost for information-state of length $q'$ is defined as
\begin{align}\label{eq.OCP q'}
    J_{q'} = \sum_{t=0}^{\bar{q}-2} c(z_t,u_t) + \sum_{t=\bar{q}-1}^{q'-2} c(z_t,u_t) + J^*_{\mathcal{Z}^{q'}},
\end{align}
where
\begin{align*}
    J^*_{\mathcal{Z}^{q'}} &= \min_{u_{q'-1}, \cdots, u_{N-1}} \sum_{t=q'-1}^{H-1} c(z_t,u_t) + c_N(z_N), \\
   \text{subject to:} ~~  \mathcal{Z}^{q'}_{t} &= \mathcal{A}^{q'}_{t-1}\mathcal{Z}^{q'}_{t-1} + \mathcal{B}^{q'}_{t-1}u_{t-1},  \text{given}\ \mathcal{Z}^{q'}_{q'-1}.
\end{align*}
Assuming $u_0, \cdots, u_{\bar{q}-2}$ are the same in Eqs. \eqref{eq.OCP q_bar} and \eqref{eq.OCP q'}, and since $u_{\bar{q}-1}, \cdots, u_{q'-2}$ are chosen optimally in Eq. \eqref{eq.OCP q_bar},  $J_{\bar{q}} \leq J_{q'}$. 
Hence, the information-state with the minimal length $\bar{q}$, results in a lower cost compared to a system with a larger initial information-state length, since the other system has to wait for additional steps to build the information-state. (If one were to use the the optimal control $u_{\bar{q}-1}, \cdots, u_{q'-2}$ for the system in eq. \eqref{eq.OCP q'}, then $J_{\bar{q}} = J_{q'}$.)
\end{proof}

\section{Identification under Noise}\label{sec.noise}
In this section, we consider the identification of the linear time-varying system under process noise through the control channel and measurement noise. We show that, given the knowledge of noise statistics, one can find an unbiased estimate of the ARMA parameters. 

Consider the following problem:
\begin{align}
    x_{t+1} &= A_{t} x_{t} + B_{t}(u_{t} + \omega_t), \nonumber \\
    z_t &= C_{t} x_{t} + \nu_t, \label{eq.system_noisy}
\end{align}
where, $\omega_t \sim \mathcal{N}(0,Q)$ and $\nu_t \sim \mathcal{N}(0,R)$. $\omega_t$ and $\nu_t$ are independent white noise sequences. To solve for the ARMA parameters without any error, one needs to solve the following linear problem:
\begin{align}\label{eq.ls_arma_noise}
    [\bar{z}^{(1)}_t \cdots \bar{z}^{(N)}_t] = [\bar{\alpha}_t ~ \bar{\beta}_t]
    \underbrace{\begin{bmatrix}
        \bar{Z}^{(1)}_{t-1} & \cdots & \bar{Z}^{(N)}_{t-1}\\
        \tilde{U}^{(1)}_{t-1} & \cdots & \tilde{U}^{(N)}_{t-1}
    \end{bmatrix}}_{\mathcal{X}_t},
\end{align}
where, $\bar{z}_t$ is the observation without noise, i.e. $\bar{z}_t = z_t - v_t$, $\tilde{u}_t = u_t + \omega_t$, $\bar{Z}_{t-1} = \tr{[\tr{\bar{z}}_{t-1}, \tr{\bar{z}}_{t-2}, \cdots, \tr{\bar{z}}_{t-q}]} $, $\tilde{U}_{t-1} = \tr{[\tr{\tilde{u}}_{t-1}, \tr{\tilde{u}}_{t-2}, \cdots, \tr{\tilde{u}}_{t-q}]}.$ Note that, one does not have access to $\bar{z}_t$ and $\tilde{u}_t$, as $v_t$ and $\omega_t$ are unknown. But, we will show that knowing $Q$, $R$, and the correlations between the input-output samples is sufficient to get the correct estimate for $\bar{\alpha}$ and $\bar{\beta}$, under the assumption that the number of samples $N$ is large. 
The result is stated below.

\begin{proposition}
The estimates of the ARMA parameters converge to their true values as the number of samples $N \to \infty$, under the assumption that the process noise acts through the control channel and the statistics of the Gaussian white noise acting on the control channel and measurements are known. 
\end{proposition}
%\textcolor{red}{you should reference the formulas for the Id in the result.}
\begin{proof}
The solution to Eq.~\eqref{eq.ls_arma_noise} is given by 
\begin{subequations}
\begin{align}
    [\hat{\alpha}_t ~ \hat{\beta}_t] &= [\bar{z}^{(1)}_t \cdots \bar{z}^{(N)}_t] \tr{\mathcal{X}}_t\inv{(\mathcal{X}_t \tr{\mathcal{X}}_t)},
 ~\text{where,} \\
     [\bar{z}^{(1)}_t \cdots \bar{z}^{(N)}_t] \tr{\mathcal{X}}_t &= N 
     \Big[R_{ZZ}(t, t-1) ~\cdots ~R_{ZZ}(t, t-q) \nonumber\\
     & R_{Z \tilde{U}}(t,t-1) \cdots R_{Z \tilde{U}}(t,t-q))\Big], \label{eq.zX} \\
     \inv{(\mathcal{X}_t \tr{\mathcal{X}}_t)} &= \frac{1}{N} 
     \begin{bmatrix}
         R_{ZZ} - I_{mq \times mq} R_{m \times m} &  R_{Z \tilde{U}} \\
         \tr{R_{Z \tilde{U}}} & R_{\tilde{U} \tilde{U}}
     \end{bmatrix} \label{eq.XX}.
 \end{align}
 \end{subequations}

 In the equations above,
 \begin{align*}
 R_{ZZ} &= \lim_{N \to \infty} \frac{1}{N} \sum_{i=1}^{N} Z_{t-1}^{(i)} Z_{t-1}^{\tr{(i)}} , \\
 R_{Z \tilde{U}} &= \lim_{N \to \infty} \frac{1}{N} \sum_{i=1}^{N} Z_{t-1}^{(i)} \tilde{U}_{t-1}^{\tr{(i)}},\\
R_{\tilde{U} \tilde{U}}  &= \lim_{N \to \infty} \frac{1}{N} \sum_{i=1}^{N} \tilde{U}_{t-1}^{(i)} \tilde{U}_{t-1}^{\tr{(i)}},\\
R_{ZZ}(t, t-l) &= \lim_{N \to \infty} \frac{1}{N} \sum_{i=1}^{N} z_{t}^{(i)} z_{t-l}^{\tr{(i)}},  \\
R_{Z \tilde{U}}(t, t-l) &= \lim_{N \to \infty} \frac{1}{N} \sum_{i=1}^{N} z_{t}^{(i)} \tr{(u_{t-l}^{(i)} + \omega_{t-l}^{(i)})}, \\
I_{mq \times mq} R_{m \times m} &= 
 \begin{bmatrix}
     R & 0_{m \times m} & \cdots & 0_{m \times m} \\
     0_{m \times m} & R & \cdots & 0_{m \times m} \\
     \vdots &  & \ddots    & \vdots \\
     0_{m \times m} & 0_{m \times m} & \cdots & R
 \end{bmatrix}.
 \end{align*}
Equations~\eqref{eq.zX} and \eqref{eq.XX} are calculated with the assumption that $N$ is large (law of large numbers).
To calculate the correlations, we use
\begin{align*}
    R_{Z \tilde{U}} &= R_{Z U} \inv{R_{UU}} (R_{UU} +I_{rq \times rq } Q), \\
    R_{\tilde{U} \tilde{U}} &= R_{UU} + I_{rq \times rq }Q, \\
    R_{Z \tilde{U}}(t, t-l) &= R_{Z U}(t,t-l) \inv{R_{UU}(t-l, t-l)} \\
                            & ~~~~~~~~~~~~~ (R_{UU}(t-l, t-l)+ Q).
\end{align*}
The quantities $R_{Z Z}, R_{Z U}$, and $R_{U U}$ are calculated from the input-output samples. 
\end{proof}
    
If the process noise acts directly on the system, the estimation of the ARMA parameters will have an unknown bias due to the unknown correlations between $z_t$ and $\omega_t$. This issue is discussed in detail here \cite{wang2023learning}.

\section{Relationship of the Information-State Approach to OKID}\label{section:okid}
%\textcolor{red}{OKID should work as advertised if the data is truly noisy owing to the innovations interpretation of the Kalman filter, we need to confirm this.}
The idea of using an ARMA model to describe the input-output data of an LTI system was first introduced in a series of papers related to the Observer/Kalman filter identification (OKID) algorithm~\cite{OKID,phan1993linearJOTA,juang1997deadbeat}, and the time-varying case was later considered in~\cite{TV_OKID}. The credit for using an ARMA model for system identification goes to the authors of the papers mentioned above, however, the explanation for the ARMA parameters given in their work is not exact, and does not apply in general as we will show empirically. This section will summarize the OKID algorithm and discuss why the information-state approach is computationally much simpler and the theory discussed in Section~\ref{section:arma_ltv} based on observability is the correct explanation for the ARMA parameters.

\subsection{The LTI Case.} \label{sec.OKID}
%\textcolor{red}{The OKID and our information-state approach both try to explain the ARMA model, and in particular, why only a finite past $q$ is enough. There are two main points now: 1) The hypothesized observer gain M in OKID turns out not to be able to explain the ARMA model, and 2) The reason you still get the correct A,B,C,D is not because of the hypothesized M, rather it is due to the development in Section IV A. The development below should proceed along these lines so that it is clear to a reader.}\todo{Edited}\\
Let us consider an LTI system:
\begin{align}
    x_{t+1} &= A x_t + B u_t, \nonumber\\
    z_t &= C x_t, \label{eq.lti_system}
\end{align}
with state $x_t \in \mathbb{R}^{n}$, input $u_t \in \mathbb{R}^{r}$ and output $z_t \in \mathbb{R}^{m}$. The input-output description for system~\eqref{eq.lti_system} can be written as a function of the past inputs, a.k.a. a moving average model, after waiting long enough for the initial condition response to die out. If the system under consideration is lightly damped, then a long history of control inputs has to be used to predict the current output accurately, as $CA^iB \approx 0, \forall i\geq q$, only for very large values of $q$.   
To mitigate this issue, the OKID approach adds a ``hypothetical" observer $M$ to the system which results in:
\begin{align}
    x_{t+1} &= A x_t + B u_t + Mz_t - Mz_t, \nonumber\\
    x_{t+1} &= \underbrace{(A + MC)}_{\bar{A}} x_t + \underbrace{[(B), -M]}_{\bar{B}} \begin{bmatrix}
    u_t \\ z_t
    \end{bmatrix}.\label{eq.lti_system_obs}
\end{align}
The input-output description of the system in Eq.~\eqref{eq.lti_system_obs} can be written as in Eq.~\eqref{eq.okid_problem}.

\begin{align}\label{eq.okid_problem}
&\underbrace{\begin{bmatrix}
z_t^{(q)}& z_t^{(q+1)}& \cdots& z_t^{(N)}
\end{bmatrix}}_{Z_t}= \nonumber\\ 
&\underbrace{\begin{bmatrix}
C\bar{B}& C\bar{A}\bar{B}& \cdots& C\bar{A}^{q-1}\bar{B}
\end{bmatrix}}_{\bar{Y}} V, 
\end{align}
where $V$ is the data matrix having all the past controls and observations stacked accordingly. For $q$, such that $mq \geq n$, the current observation can be predicted from the past $q$ observations and controls, assuming $C\bar{A}^i\bar{B} = 0, \forall\ i\geq q$. Adding an observer moves the poles of $\bar{A}$ which dampens the system and greatly reduces the number of past inputs and observations needed. \\
\indent The least-squares solution to estimate $\bar{Y}$ - ``Observer Markov parameters" (as defined in \cite{OKID}) of this observer system is given by $\bar{Y} = Z_t V^{\dagger}$. The open-loop Markov parameters, $CA^tB$, are then calculated from $\bar{Y}$ (Eq. 20 in \cite{OKID}), which are then used to find $\hat{A},\hat{B},\hat{C}$ using the eigen realization algorithm (ERA)~\cite{juang1985eigensystem}. The observer $M$ is then calculated using
$
\hat{M} = \inv{(\tr{\mathcal{O}}\mathcal{O})}\tr{\mathcal{O}}Y^o,
$
where, $\mathcal{O}$ is the observability matrix, constructed using $\hat{A}, \hat{C}$ and $Y^o$ is a matrix containing the observer gain Markov parameters calculated using $\bar{Y}$ as defined in \cite{OKID}.\\ 
\indent OKID starts with the hypothesis that the ARMA parameters can be explained using the observer system~\eqref{eq.lti_system_obs}, in particular that the identified linear system $\hat{A}, \hat{B}, \hat{C}$, and the identified observer $\hat{M}$, explain the identified ARMA parameters, i.e., $\hat{C}(\hat{A} + \hat{M}\hat{C})^i[\hat{B},-\hat{M}] = \bar{Y}_i$. However, in general, the reconstructed observer Markov parameters using the identified $\hat{A}, \hat{B}, \hat{C}, \hat{M}$ don't match the ARMA parameters estimated from Eq.~\eqref{eq.okid_problem}, as shown in Fig.~\ref{fig:okid_sys}.  

\begin{remark}
The special case in which OKID's observer Markov parameters match the ARMA parameters is when the observer is deadbeat (Ch. 9.3.4 \cite{Antsaklis97}), i.e. eigenvalues of $\hat{A}+\hat{M}\hat{C}$ are at the origin, which results in $(\hat{A}+\hat{M}\hat{C})^q=0$. But the deadbeat condition is not satisfied in the general case.
\end{remark}

\begin{comment}
\indent \todo{The problem with OKID is the assumption that a deadbeat observer can be calculated for any general case (considering only deterministic systems here). Our experiments (see Fig.~\ref{fig:okid_sys}) revealed that, the deadbeat observer condition can be satisfied only in some special cases, and for a specific value of $q$. In most cases, the calculated $M$ is not a deadbeat observer, meaning, the deadbeat observer is not the explanation for the ARMA model which is exact in modeling the input-output response. The real reason why ARMA model fits the input-output data exactly (after $q$ steps) is because of observability as opposed to having an observer. As shown in section~\ref{section:arma_ltv}, the ARMA parameters reconstruct the initial condition implicitly, which lets it model both the initial condition response and the forced response exactly. }
\end{comment}
To summarize, OKID and our information-state approach both try to explain the ARMA model, and in particular, why only a finite past $q$ steps are enough. OKID tries to justify that the ARMA parameters can be modeled as an observer system, which our experiments contradict by showing that the observer doesn't accurately explain the ARMA parameters. While, the information-state approach, as discussed in Sec.~\ref{section:arma_ltv}, uses observability to implicitly reconstruct the state $q$ steps in the past, which helps it model both the initial condition response as well as the forced response exactly after an initial transient of $q$ steps. The reason the open-loop Markov parameters $CA^t B$ match the true Markov parameters in OKID, is not because of the hypothesized observer $M$, rather it is due to the development in the following Sec.~\ref{sec.deriv_markov}, which is based on the theory discussed in Sec.~\ref{section:arma_ltv}. Finally, the construction of the information-state model doesn't need any additional steps beyond the calculation of the ARMA parameters whereas OKID has to calculate the open-loop Markov parameters, calculate $\hat{A}, \hat{B}, \hat{C},$ by doing an SVD of the Hankel matrix, and then calculate $\hat{M}$ to construct the observer system.
\begin{figure}[!htbp]
\centering
   {\includegraphics[width=0.8\linewidth]{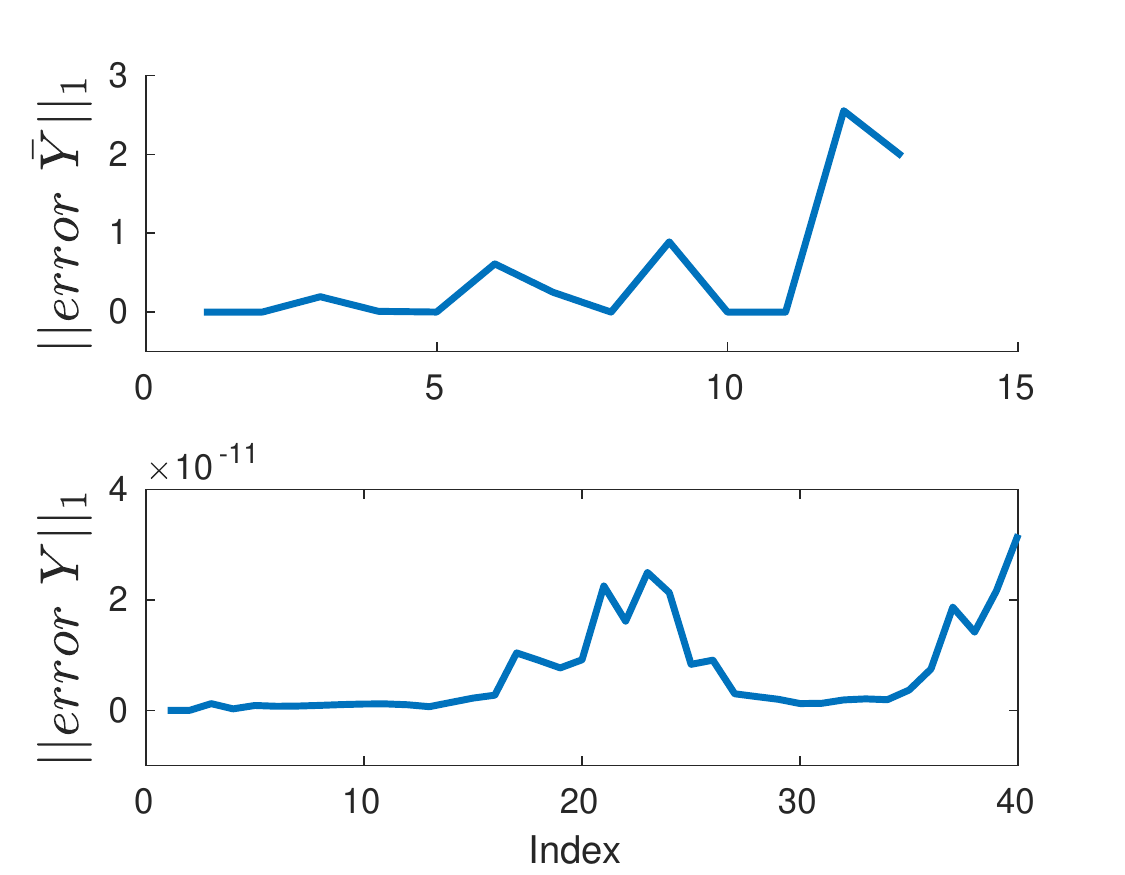} }
    \caption{The errors in the observer Markov parameters $\bar{Y}$ for $q=4$, and the open-loop Markov parameters $Y$ calculated on a 3-DoF Spring Mass system used in \cite{OKID} (system details are given in Table \ref{tab:sims}). The plotted errors are normalized by the true values and the 1-norm of the two channels is calculated. Error $Y$ is the error between the true open-loop Markov parameters of the system and the calculated open-loop parameters from the OKID algorithm. The error $\bar{Y}$ is the error between the estimated ARMA parameters ($\bar{Y} = Z_t V^{\dagger}$) and the observer Markov parameters reconstructed from the $\hat{A},\hat{B},\hat{C},\hat{M}$ resulting from the OKID algorithm. The takeaway from this experiment is that the open-loop Markov parameters and the calculated $\hat{A},\hat{B},\hat{C}$, match the true system, but the calculated observer $\hat{M}$ doesn't explain the ARMA parameters, as the reconstructed observer Markov parameters using $\hat{M}$ don't match the ARMA parameters. We also see that the reconstructed $C\bar{A}^q\bar{B}\neq0$ from the experiment, invalidates the deadbeat assumption.}
    \label{fig:okid_sys}
\end{figure}

\subsection{Calculation of the open-loop Markov parameters from the ARMA parameters}\label{sec.deriv_markov}
Restating Eq.~\eqref{eq.true_ls_sol} for the LTI case:
\begin{align}
    [\alpha_{1} \cdots \alpha_{q}] &= CA^q(O^q)^{\dagger} \nonumber \\
    [\beta_{1} \cdots \beta_{q}] &= ([C B, \cdots, C A^{q-1}B] -CA^q(O^q)^{\dagger} G^q) \label{eq.true_ls_sol_lti}
\end{align}
Substituting the first equation into the second and rearranging, we get, 
\begin{align}
    [C B, \cdots, C A^{q-1}B] = [\beta_{1} \cdots \beta_{q}] + [\alpha_{1} \cdots \alpha_{q}] G^q.
\end{align}
Then, 
\begin{align}
    Y_0 := CB &= \beta_1, \nonumber\\
    Y_1 : = CAB & = \beta_2 + \alpha_1 Y_0, \nonumber\\
    Y_{t-1} := CA^{k-1}B &= \beta_t +\sum_{i=1}^{k-1} \alpha_i Y_{t-i-1}.
\end{align}
The Markov parameters can be calculated to any length by setting $\alpha_t$ and $\beta_t$ to $0,\ \forall\ k > q$. The equations used to derive the open-loop Markov parameters from the ARMA parameters are identical to OKID, but they originate from Eq.~\eqref{eq.output_model} and~\eqref{eq.true_ls_sol}, which do not require OKID's hypothesized observer $M$.
%\textcolor{red}{The idea is that there is a contradiction within the OKID theory, the hypothesized observer $M$ .}
%\todo{Edited}

\begin{table}[!htbp]
    \centering
    \resizebox{\columnwidth}{!}{%
    \begin{tabular}{|c|c|c|}
        \hline    
          \textbf{Requirements/Issues}& \textbf{TV-OKID} & \textbf{Information-state}  \\
         \hline
         ARMA parameters & Yes & Yes \\
         \hline
         Observability & Yes & Yes\\
         \hline
         $mq \geq n$ & Yes & Yes\\
         \hline
         $rq \geq n$ & Yes & No\\
         \hline
         Computing open-loop Markov & &\\
         parameters and Hankel & Yes & No \\
         \hline
         Experiments to be performed & & \\
         with zero-initial condition &Yes & No\\
         \hline 
         Free-response experiment & Yes &No \\
         \hline
         Coordinate transformations & Yes & No\\
         \hline
         Issue in calculating final $q$ steps & Yes & No\\
         \hline
    \end{tabular}}
    \caption{Difference between TV-OKID and Information-state model.}
    \label{tab:diff_okid_arma}
\end{table}
\subsection{LTV case - Time-varying Observer/Kalman Filter Identification:}\label{sec.TV_OKID}
The time-varying OKID algorithm (TV-OKID)~\cite{TV_OKID} generalizes the OKID approach by estimating the time-varying ARMA parameters, and goes on to calculate the open-loop Markov parameters from the TV-ARMA parameters. Using the open-loop Markov parameters, ERA is used to calculate $\hat{A}_t,\hat{B}_t,\hat{C}_t$, and the observer $\hat{M}_t = \mathcal{\hat{O}}_{t+1}^{\dagger}P_{t+1}$, where $P_{t+1}$ is a matrix containing the observer gain Markov parameters. The hypothesis is that the calculated system and observer fit the TV-ARMA parameters. The TV-OKID algorithm is complex compared to the time-invariant case, as the calculated $\hat{A}_t,\hat{B}_t,\hat{C}_t,\hat{M}_t$ matrices are in a different coordinate system at every time-step $t$. A coordinate transformation has to be performed to bring them to a reference coordinate system~\cite{TV_ERA}. The same issues discussed in the time-invariant case were seen in the time-varying case as well. The Markov parameters of the observer system do not match the TV-ARMA parameters. In addition to needing more computational steps to construct the model, it also requires the experiments to be carried out strictly from zero-initial conditions. In addition to that, another free-response experiment, which involves collecting responses of the system from random non-zero initial conditions with zero forcing, has to be performed to identify models for the first few time-steps and the last few time-steps. Finally, the last $q$ steps cannot be identified for finite-time problems, as the response has to be recorded for $q$ steps after the final time interval. \\
\indent The Information-state based model doesn't suffer from any of the above mentioned issues. It is simple, as the state-space model can be arrived at, immediately after calculating the TV-ARMA parameters, courtesy of the information-state. On the other hand, the TV-OKID has to compute the open-loop Markov parameters, build the Hankel matrix to calculate $\hat{A}_t, \hat{B}_t, \hat{C}_t$, and then transform them to a common reference coordinate system. For prediction, the initial condition has to be explicitly reconstructed after collecting observations in the first $q$ steps, for which it relies on observability - which the information-state model is inherently built upon. The additional requirements and issues pertaining to TV-OKID are tabulated in Table~\ref{tab:diff_okid_arma}.\\
\indent The objective of OKID/TV-OKID to find the open-loop parameters from the ARMA parameters might have been influenced from the traditional Kalman-Ho realization/ ERA procedure and the conventional definition of the state. In essence, the open-loop parameters model the impulse response/forced response, while, the ARMA captures both the initial condition response and the forced response, after an initial transient. The method used to calculate the open-loop parameters from ARMA parameters, strips away the initial condition response to arrive at the impulse response parameters. This is overcome in our method, by using the information-state and constructing the state-space model directly from the ARMA parameters.

\section{Empirical Results}\label{section:results}
%\textcolor{red}{This section needs a clearer explanation.} \todo{Edited}
In the following, we show the performance of the information state based LTV identification technique on a simple LTI system as well as complex finite time LTV problems obtained by linearizing nonlinear models around a nominal trajectory (Section~\ref{sec.id_exp}). We also show, in Section~\ref{sec.equivalence_exp}, the optimality of the information state feedback control as verification of the theory proposed in Section~\ref{sec.control_implications}. 
\subsection{Performance of the information-state based system identification}\label{sec.id_exp}
We tested the information-state technique on the oscillator system used in \cite{TV_ERA, TV_OKID} and on the cart-pole and fish systems available in the open-source MuJoCo simulator~\cite{mujoco}. The details of the systems are given in Table~\ref{tab:sims}. While the oscillator system is a true LTV system, the cart-pole and the fish are nonlinear systems. \\
\begin{figure}[!htbp]
\centering
    \subfloat{\includegraphics[width=0.48\linewidth]{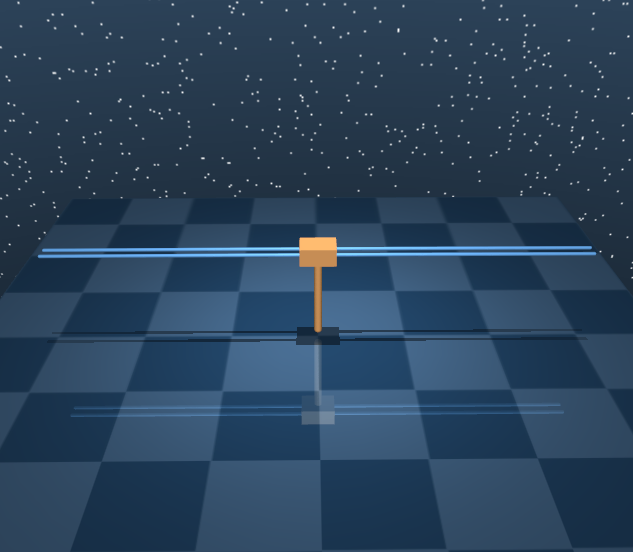}}
      \subfloat{\includegraphics[width=0.48\linewidth]{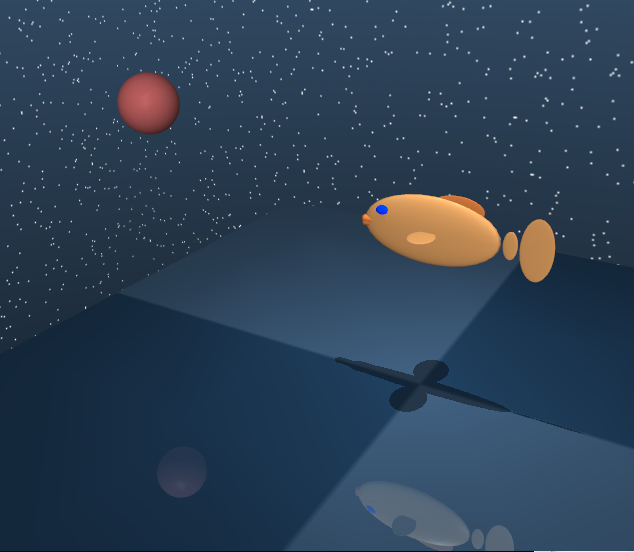}}
\caption{Models simulated in MuJoCo.}
\label{fig:mujoco_models}
\end{figure}
\begin{table}[!htbp]
    \centering
    \resizebox{\columnwidth}{!}{%
    \begin{tabular}{|c|c|c|c|c|c|}
    \hline
         \textbf{System}& \textbf{Horizon} & $q$ & \textbf{State} & \textbf{Output} & \textbf{Input}\\
         & &  & \textbf{dim.} ($n$) & \textbf{dim.} ($m$) & \textbf{dim.} ($r$)  \\
         \hline
         3 DoF Spring- &  &  & && \\ 
         mass (LTI) \cite{OKID} & 40 & 4 & 6 & 2 & 1\\
         \hline
         Oscillator \cite{TV_ERA} & 30 & 4 & 4 & 2 & 2\\
         \hline
         Cart-pole & 31 & 4 & 4 & 2& 1\\
         \hline
         Fish & 30 & 5 & 27 & 11 & 6\\
         \hline
    \end{tabular}}
    \caption{Simulation details.}
    \label{tab:sims}
\end{table}
\begin{figure}[!htbp]
\centering
    \subfloat[Cart-pole]{\includegraphics[width=0.48\linewidth]{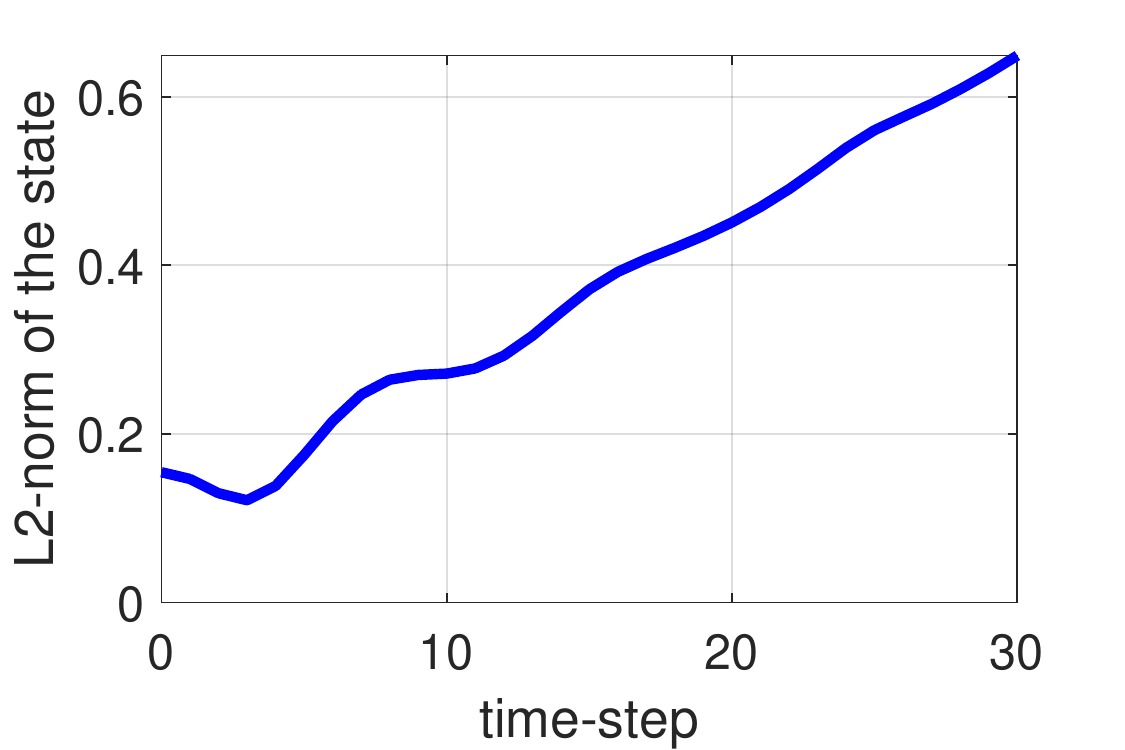}}
      \subfloat[Fish]{\includegraphics[width=0.48\linewidth]{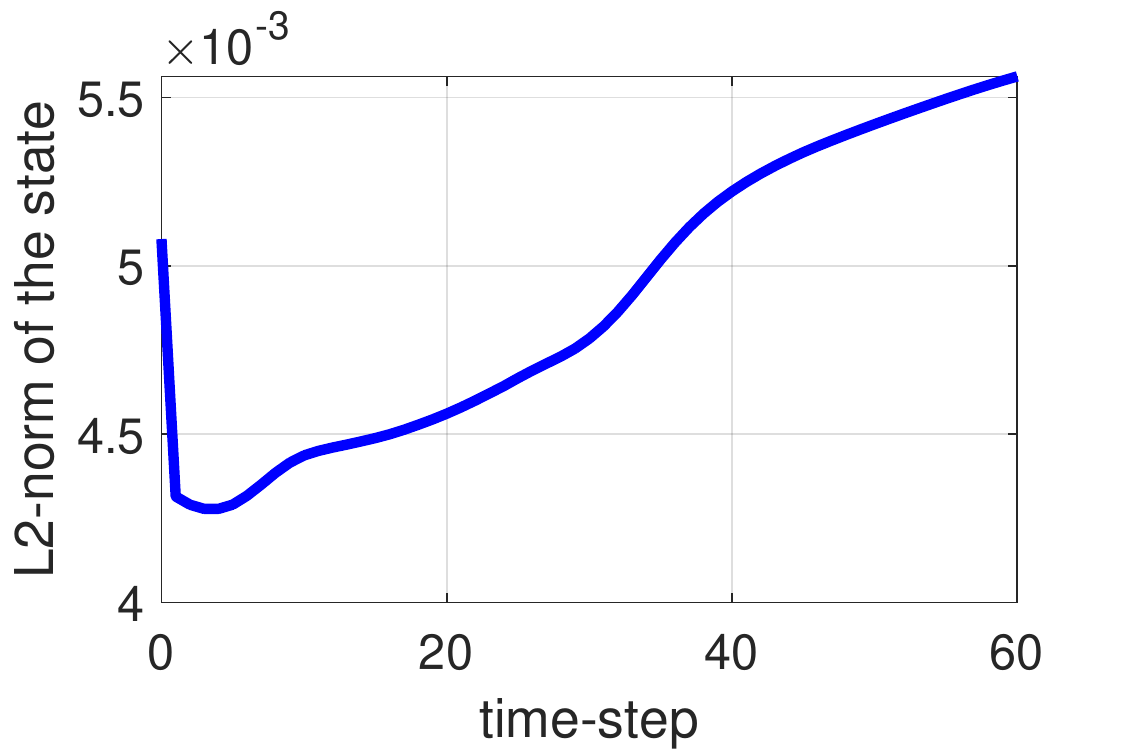}}
\caption{An illustration of the fact that non-zero initial conditions don't decay to 0, in general. We take non-zero initial conditions and let them evolve with time with only the nominal control. The plot shows that the error due to the non-zero initial conditions doesn't decay to 0 in finite time. It justifies the importance of having a system identification technique that doesn't require the initial conditions to decay to 0.} 
\label{fig:initial_condition}
\end{figure}
\indent For the oscillator, the identification is straightforward. The system was perturbed with random inputs from a normal distribution and the corresponding output responses were recorded.  For nonlinear systems, system identification has to be carried out along a trajectory about which the system can be linearized. Given a nonlinear system of the form $x_{t+1} = F(x_t,u_t),\ z_t = C_t x_t$, it can be linearized along a trajectory $\{ \bar{x}_t, \bar{u}_t \}^{T}_{t=0}$. The resulting LTV system is given by: $\delta x_{t+1} \approx A_t \delta x_{t} + B_t \delta u_t$, where $\delta x_t = x_t - \bar{x}_t$, $\delta u_t = u_t - \bar{u}_t$ and $A_t = \frac{\partial F}{\partial x}|_{\bar{x},\bar{u}}$, $B_t = \frac{\partial F}{\partial u}|_{\bar{x},\bar{u}}$. And, $\delta z_t = C_t \delta x_t.$ The nominal trajectory was computed using the iLQR algorithm~\cite{POD2C} and the perturbed system is identified around that trajectory using the input-output responses for every $t$ from $N$ independent experiments/rollouts $\{\delta u_t^{(1:N)}, \delta z_t^{(1:N)} \}$. \\
\indent Using the input-output responses from the experiments, the TV-ARMA parameters were estimated, from which the information-state based state-space model (Eq.~\eqref{eq.info_state_model}) is constructed. TV-OKID was also implemented to compare its performance. In addition to the forced response experiment, TV-OKID requires a free-response experiment. So, a set of random initial conditions are taken for the system under consideration, and the unforced responses are recorded to identify the system matrices for the first $q$ steps. TV-OKID algorithm computes the system matrices $\hat{A}_t,\hat{B}_t,\hat{C}_t,\hat{M}_t$ from the TV-ARMA parameters and the free-response experiment, using the procedure detailed in~\cite{TV_OKID}. The initial state is estimated using the observations from the first $q$ steps. 
%We also plot the response of the system without including the observer (labeled TV-OKID (without obs.)), to show its performance. 

First, we show the importance of having a system identification technique that is immune to non-zero initial conditions. We show in Fig. \ref{fig:initial_condition} that the non-zero initial conditions, in general, don't decay to zero in finite time. Next, we discuss the performance comparison of TV-OKID and the Information-state approach. For the oscillator, two experiments were performed - one with zero-initial conditions and another with non-zero initial conditions, and the results are shown in Fig.~\ref{fig:oscillator}. The results for the nonlinear systems are shown in Fig~\ref{fig:mujoco_sims}. The error shown in the figures is the  1-norm of the mean error between the true response and the predicted response from 100 independent simulations, across all the output channels. \\
\begin{figure}[!htbp]
    \centering
   \subfloat[Model calculated using \\responses from zero \\initial conditions]{\includegraphics[width=0.48\linewidth]{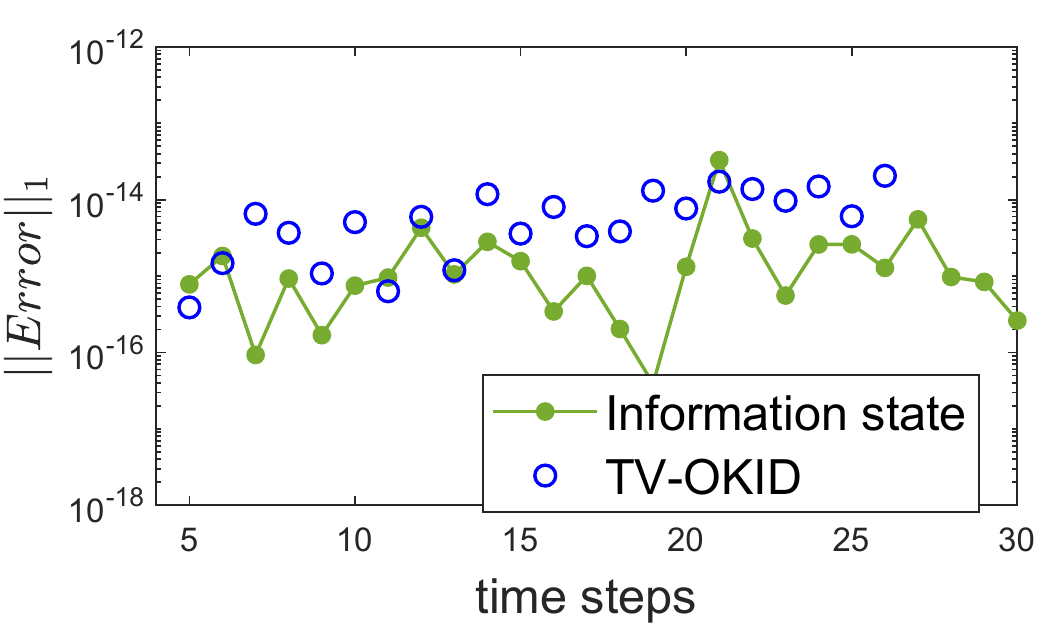} \label{fig:oscillator_q_4_zeroini}}
   %\newline
   %\centering
    \subfloat[Model calculated using \\ responses from non-zero \\initial conditions] {\includegraphics[width=0.48\linewidth]{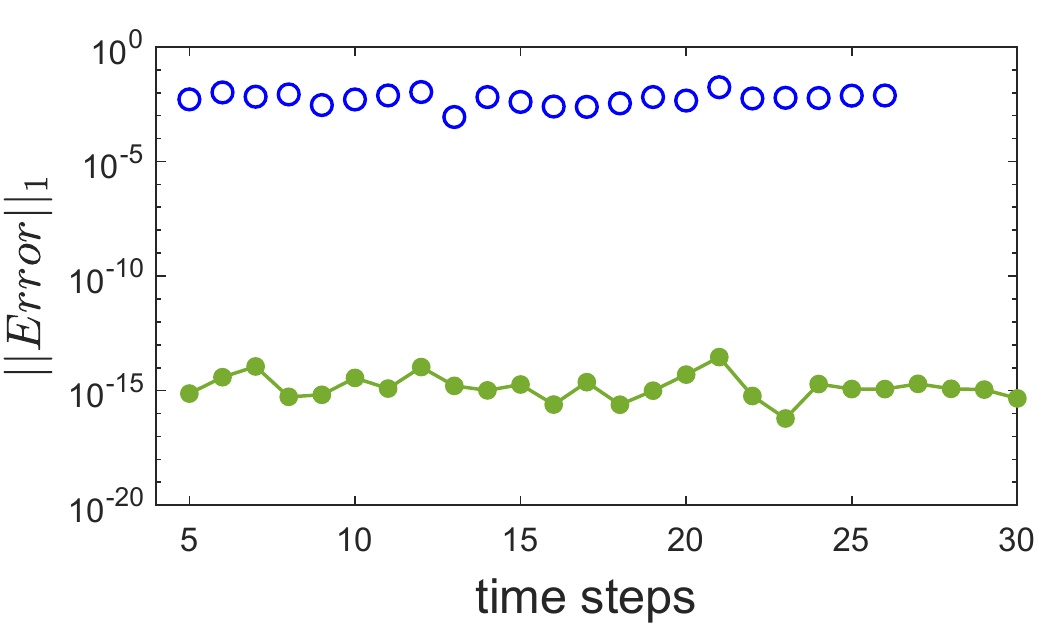} \label{fig:oscillator_q_4_nonzero_ini}}%
    \caption{Results on the oscillator model. OKID requires the data used for identification to have zero-initial conditions, to calculate the first few Markov parameters accurately. 
    %The plot `TV-OKID (without obs.)' uses the identified $\hat{A}_t,\hat{B}_t,\hat{C}_t,\hat{D}_t$ matrices for prediction, without the observer $\hat{M}_t$ in the loop.
    }
    \label{fig:oscillator}
\end{figure}
\begin{figure}[!htbp]
   \subfloat[Cart-pole with zero \\initial conditions]{\includegraphics[width=0.48\linewidth]{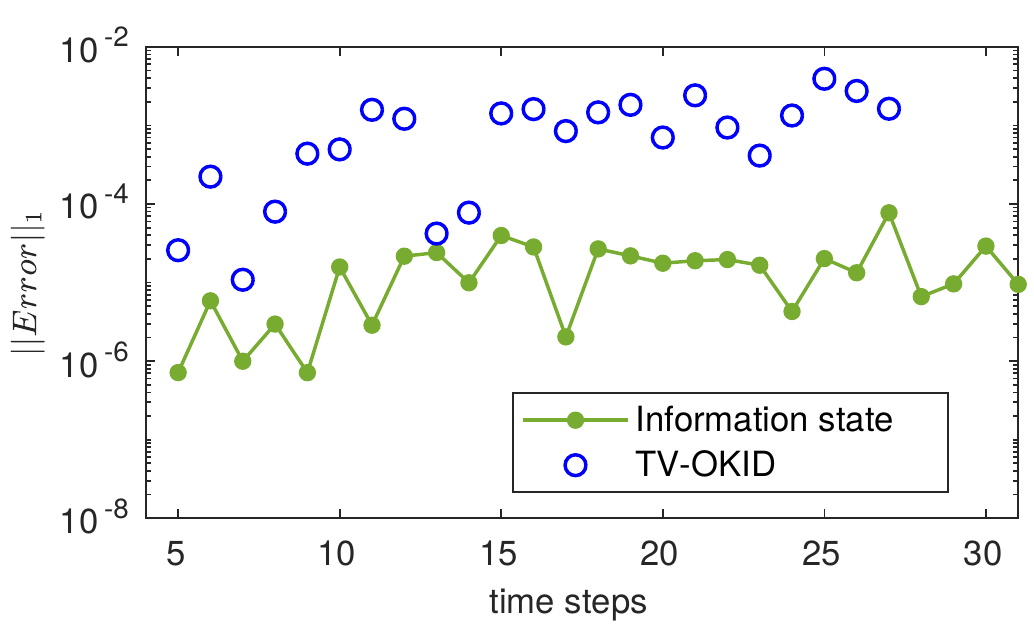} \label{fig:cartpole}}%
    \subfloat[Cart-pole with non-zero \\initial conditions]{\includegraphics[width=0.48\linewidth]{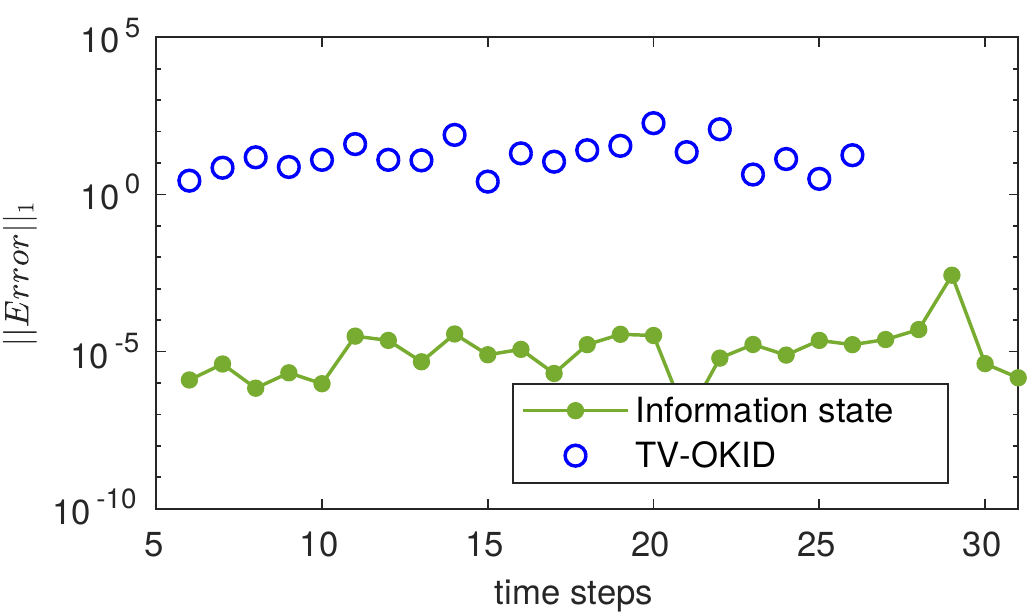}\label{fig:cartpole1}}%
    \newline
    \subfloat[Fish with zero \\initial conditions] {\includegraphics[width=0.48\linewidth]{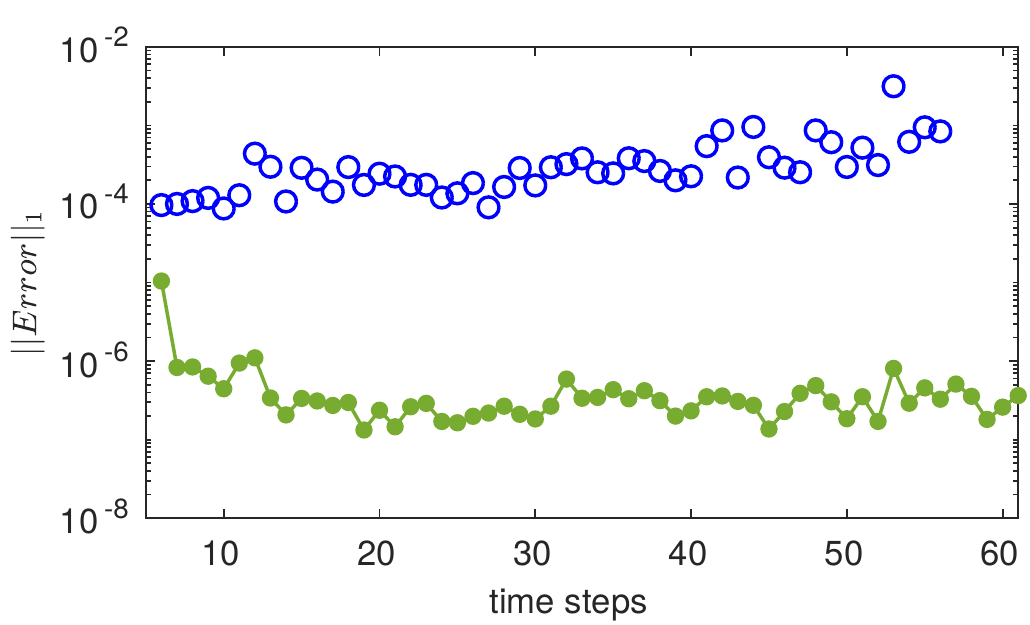} \label{fig:fish}}%
     \subfloat[Fish with non-zero \\initial conditions] {\includegraphics[width=0.48\linewidth]{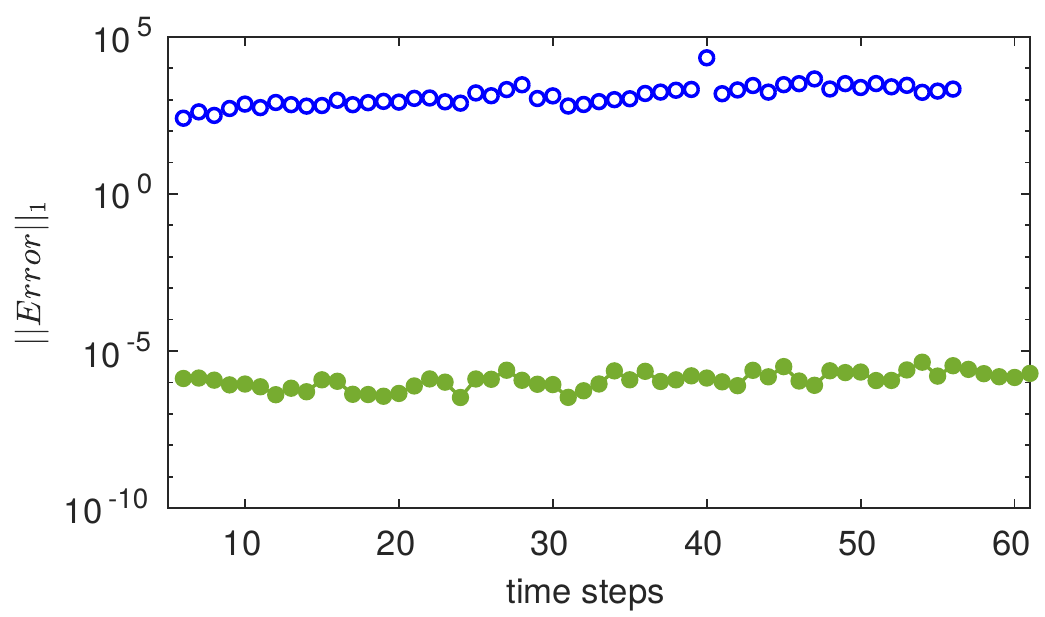} \label{fig:fish1}}%
    \caption{Results on the nonlinear models. The experiments for identifying the system were performed from zero-initial conditions and non-zero initial conditions. The identified observer in both cases is sensitive. Even small errors between the true response and the predicted response in the first few steps are amplified, leading to instability. Hence, we plot only the TV-OKID without the observer.}
    \label{fig:mujoco_sims}
\end{figure}
The results show that the information-state model can predict the responses accurately. The TV-OKID approach also can predict the response well in the oscillator experiment when the experiments have zero initial conditions, but it suffers from inaccuracy if the experiments have non-zero initial conditions as seen in Fig.~\ref{fig:oscillator_q_4_nonzero_ini}. In the case of fish and cart-pole, TV-OKID fails with the observer in the loop. We found that the identified open-loop Markov parameters predict the response well, but the prediction diverges from the truth when the observer is introduced, making the predictions useless. This observation further validates the hypothesis that the ARMA model cannot be explained by an observer in the loop system. Hence, we use only the estimated open-loop Markov parameters without the observer to show the performance of the TV-OKID prediction. The last $q$ steps in OKID are ignored, as there is not sufficient data to calculate models for the last few steps, as discussed in Sec.~\ref{sec.TV_OKID}. There is also the potential for numerical errors to creep in due to the additional steps taken in TV-OKID: determination of the time-varying Markov parameters from the time-varying observer Markov parameters, calculating the SVD of the resulting Hankel matrices and the calculation of the system matrices from these SVDs, as mentioned in~\cite{TV_OKID}. On the other hand, the effort required to identify systems using the information-state approach is negligible compared to other techniques as the state-space model can be set up by just using the ARMA parameters. More examples can be found in \cite{Ran_ICRA21}, where the authors use the information-state model for optimal feedback control synthesis in complex nonlinear systems.

\subsection{Optimality of feedback control using information-state}\label{sec.equivalence_exp}
In the following, we present empirical results to show the optimality of feedback control using information-state system as discussed in Section~\ref{sec.control_implications}. We show here that the optimal control using information-state system is exactly the same as the optimal control using the true state-space system, as proved in Theorem \ref{theorem.optimal_fb}. In Fig. \ref{fig:fb_control}, we show results on the 3 DoF spring-mass system. The optimal control problem is to take the system from a non-zero initial condition to the origin by minimizing a quadratic cost function - $\min_{\{u_t\}} \sum_{t=0}^{H-1} \frac{1}{2}(\tr{z}_t Q z_t + \tr{u}_t R u_t) + \frac{1}{2}\tr{z}_T Q_T z_T$. The minimal value $\bar{q}$ is taken for the information-state system and the first $\bar{q}$ control inputs are taken arbitrarily to construct the initial state for both systems. A finite horizon linear quadratic regulator is designed for both systems for a horizon of $50$ time steps. The control inputs and the responses of both the systems are identical as shown by the negligible errors between the two in Fig.~\ref{fig:fb_control}, supporting our theory that the systems are equivalent and generate the same optimal control inputs.

\begin{figure}
    \centering
    \includegraphics[width =0.7\linewidth]{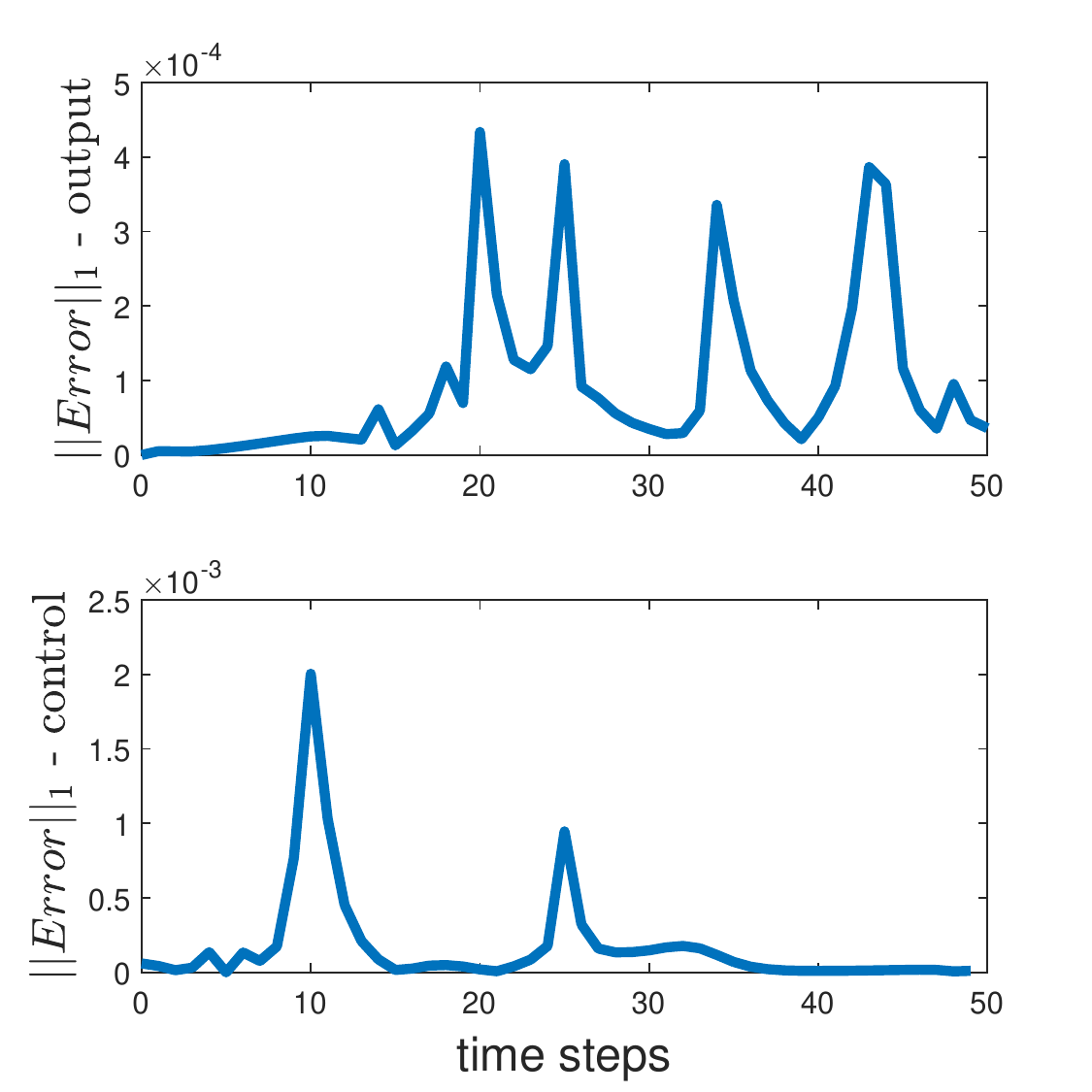}
    \caption{Optimal feedback control on the 3 DoF spring-mass (LTI) system. The error plotted is the $L_1$ error between the response of the information-state system and the true system, normalized by the true system's response. The difference in cost between the information-state controller and the state-space controller is $3.5666 \times 10^{-8}$ times the cost of the state-space controller, which is negligible.}
    \label{fig:fb_control}
\end{figure}

\section{Conclusion}\label{section:Conclusion}
This paper describes a new system realization technique for the system identification of linear time-invariant as well as time-varying systems. The system identification method proceeds by modeling the current output of the system using an ARMA model comprising of the finite past outputs and inputs. A theory based on linear observability is developed to justify the usage of an ARMA model, which also provides the minimum number of inputs and outputs required from history for the model to fit the data exactly. The method uses the information-state, which simply comprises of the finite past inputs and outputs, to realize a state-space model directly from the ARMA parameters. This is shown to be universal for both linear time-invariant and time-varying systems that satisfy the observability assumption. Further, we show that feedback control based on the minimal information state is optimal for the underlying state space system, i.e., the information state is indeed a loss-less representation for the purpose of control.  The method is tested on various systems in simulation, and the results show that the models are accurately identified. 

\printbibliography

@article{TV_ERA,
author = {Majji, M. and Juang, J. and Junkins, J. L.},
title = {Time-Varying Eigensystem Realization Algorithm},
journal = {Journal of Guidance, Control, and Dynamics},
volume = {33},
number = {1},
pages = {13-28},
year = {2010},
doi = {10.2514/1.45722},

URL = { 
        https://doi.org/10.2514/1.45722
},
}

@article{TV_OKID,
author = {Majji, Manoranjan and Juang, Jer-Nan and Junkins, John L.},
title = {Observer/Kalman-Filter Time-Varying System Identification},
journal = {Journal of Guidance, Control, and Dynamics},
volume = {33},
number = {3},
pages = {887-900},
year = {2010},
doi = {10.2514/1.45768},

URL = { 
        https://doi.org/10.2514/1.45768
    
},
}

@article{OKID,
author = {Juang, Jer-Nan and Phan, Minh and Horta, Lucas G. and Longman, Richard W.},
title = {Identification of observer/Kalman filter Markov parameters - Theory and experiments},
journal = {Journal of Guidance, Control, and Dynamics},
volume = {16},
number = {2},
pages = {320-329},
year = {1993},
doi = {10.2514/3.21006},

URL = { 
        https://doi.org/10.2514/3.21006
    
},
}

@INPROCEEDINGS{Ran_ICRA21,
  author={Wang, Ran and Goyal, Raman and Chakravorty, Suman and Skelton, Robert E.},
  booktitle={IEEE International Conference on Robotics and Automation (ICRA)}, 
  title={Data-based Control of Partially-Observed Robotic Systems}, 
  year={2021},
  volume={},
  number={},
  pages={8104-8110},
  doi={10.1109/ICRA48506.2021.9561001}}

@INPROCEEDINGS{mujoco,
  author={Todorov, Emanuel and Erez, Tom and Tassa, Yuval},
  booktitle={2012 IEEE/RSJ International Conference on Intelligent Robots and Systems}, 
  title={MuJoCo: A physics engine for model-based control}, 
  year={2012},
  volume={},
  number={},
  pages={5026-5033},
  doi={10.1109/IROS.2012.6386109}}

@article{POD2C,
      title={An Information-state based Approach to the Optimal Output Feedback Control of Nonlinear Systems}, 
      author={Raman Goyal and Ran Wang and Mohamed Naveed Gul Mohamed and Aayushman Sharma and Suman Chakravorty},
      year={2023},
      eprint={2107.08086},
      archivePrefix={arXiv},
      primaryClass={cs.RO}
}

@article{phan1993linearJOTA,
  title={Linear system identification via an asymptotically stable observer},
  author={Phan, Minh and Horta, Lucas G and Juang, Jer-Nan and Longman, Richard W},
  journal={Journal of Optimization Theory and Applications},
  volume={79},
  number={1},
  pages={59--86},
  year={1993},
  publisher={Springer}
}

@book{Antsaklis97,
  address = {New York},
  author = {Antsaklis, Panos J. and Michel, Anthony N.},
  description = {robotica-bib},
  publisher = {McGraw-Hill,},
  title = {Linear systems },
  year = 1997
}

@article{HOKALMAN,
url = {https://doi.org/10.1524/auto.1966.14.112.545},
title = {Effective construction of linear state-variable models from input/output functions},
journal = {Proceedings of the 3rd
Annual Allerton Conference on Circuit and System Theory},
author = {Ho, B. L. and Kalman, R. E.},
pages = {545--548},
volume = {14},
number = {1-12},
doi = {doi:10.1524/auto.1966.14.112.545},
year = {1965},
lastchecked = {2022-10-02}
}

@incollection{ljung1998system,
  title={System identification},
  author={Ljung, Lennart},
  booktitle={Signal analysis and prediction},
  pages={163--173},
  year={1998},
  publisher={Springer}
}

@book{juang1994applied,
  title={Applied system identification},
  author={Juang, Jer-Nan},
  year={1994},
  publisher={Prentice-Hall, Inc.}
}

@article{juang1985eigensystem,
  title={An eigensystem realization algorithm for modal parameter identification and model reduction},
  author={Juang, Jer-Nan and Pappa, Richard S},
  journal={Journal of guidance, control, and dynamics},
  volume={8},
  number={5},
  pages={620--627},
  year={1985}
}

@inproceedings{majji2018timeCOVER,
  title={Time Varying Covariance Equivalent Realizations},
  author={Majji, Manoranjan},
  booktitle={American Control Conference},
  pages={283--287},
  year={2018},
  organization={}
}

@article{king1988generalized,
  title={A generalized approach to $q$-Markov covariance equivalent realizations for discrete systems},
  author={King, Andrew M and Desai, Uday B and Skelton, Robert E},
  journal={Automatica},
  volume={24},
  number={4},
  pages={507--515},
  year={1988},
  publisher={Elsevier}
}

@article{shokoohi1987identification,
  title={Identification and model reduction of time-varying discrete-time systems},
  author={Shokoohi, Shahriar and Silverman, Leonard M},
  journal={Automatica},
  volume={23},
  number={4},
  pages={509--521},
  year={1987},
  publisher={Elsevier}
}

@techreport{juang1997deadbeat,
  title={Deadbeat predictive controllers},
  author={Juang, Jer-Nan and Phan, Minh},
  year={1997}
}

@article{VERHAEGEN1995201,
title = {A class of subspace model identification algorithms to identify periodically and arbitrarily time-varying systems},
journal = {Automatica},
volume = {31},
number = {2},
pages = {201-216},
year = {1995},
issn = {0005-1098},
doi = {https://doi.org/10.1016/0005-1098(94)00091-V},
url = {https://www.sciencedirect.com/science/article/pii/000510989400091V},
author = {Michel Verhaegen and Xiaode Yu}
}

@book{Bert05,
  title                    = {Dynamic Programming and Optimal Control},
  author                   = {Dimitri P. Bertsekas},
  publisher                = {Athena Scientific},
  year                     = {2000},
  edition                  = {2nd},
  volume                   = {I},
}

@misc{kumar1986stochastic,
  title={Stochastic systems: estimation, identification and adaptive control},
  author={Kumar, PR and Varaiya, Pravin},
  year={1986},
  publisher={Prentice-Hall, Inc.}
}

@INPROCEEDINGS{informationstateACC,
  author={Mohamed, Mohamed Naveed Gul and Goyal, Raman and Chakravorty, Suman and Wang, Ran},
  booktitle={American Control Conference}, 
  title={The Information-State Based Approach to Linear System Identification}, 
  year={2023},
  volume={},
  number={},
  pages={301-306},
  doi={10.23919/ACC55779.2023.10156137}}

@INPROCEEDINGS{wang2023learning,
      title={Learning to Control under Uncertainty with Data-Based Iterative Linear Quadratic Regulator}, 
      author={Ran Wang and Raman Goyal and Suman Chakravorty},
      year={2023},
      eprint={2311.04852},
      archivePrefix={arXiv},
      primaryClass={cs.RO}
}
\end{document}